\documentclass[aps,prd,nofootinbib,preprint,tightenlines]{revtex4-1}

\usepackage{todonotes,soul}
\usepackage{epsfig}
\usepackage{braket}
\usepackage{graphicx}
\usepackage{amsmath}
\usepackage{amsthm}
\usepackage{amssymb}
\usepackage{amsfonts}
\usepackage{mathbbol}
\usepackage{caption}
\usepackage{bm}
\usepackage{xcolor}
\usepackage{mathrsfs}

\newcommand{\dV}{d\mathrm{Vol}}
\newcommand{\RN}[1]{%
	\textup{\uppercase\expandafter{\romannumeral#1}}%
}

\newcommand{\be}{\begin{equation}}
\newcommand{\ee}{\end{equation}}
\newcommand{\bea}{\begin{eqnarray}}
\newcommand{\eea}{\end{eqnarray}}
\def\bml{\begin{subequations}}
\def\blea{\bml\begin{eqnarray}}
\def\eml{\end{subequations}}
\def\elea{\end{eqnarray}\eml}

\DeclareMathOperator{\WF}{WF}

\newcommand{\coin}[1]{\left[\!\!\left[#1\right]\!\!\right]}

\renewcommand{\sec}{Sec$.$~}

\newtheorem{theorem}{Theorem}

\newtheorem{lemma}[theorem]{Lemma}
\newtheorem{definition}[theorem]{Definition}

\newcommand{\nord}[1]{{:}#1{:}}

\begin{document}
\title{Quantum strong energy inequalities}
 
\author{Christopher J. Fewster}\email{chris.fewster@york.ac.uk}
\author{Eleni-Alexandra Kontou}\email{eleni.kontou@york.ac.uk}
\affiliation{Department of Mathematics, University of York, Heslington, York YO10 5DD, United Kingdom}
\date{\today}

\begin{abstract}
Quantum energy inequalities (QEIs) express restrictions on the extent to which weighted averages of the renormalized energy density can take negative expectation values within a quantum field theory. Here we derive, for the first time, QEIs for the effective energy density (EED) for the quantized non-minimally coupled massive scalar field. The EED is the quantity required to be non-negative in the strong energy condition (SEC), which is used as a hypothesis of the Hawking singularity theorem. Thus  establishing such quantum strong energy inequalities (QSEIs) is a first step towards a singularity theorem for matter described by quantum field theory. 

More specifically, we derive difference QSEIs, in which the local average of the EED is normal-ordered relative to a reference state, and averaging occurs over both timelike geodesics and spacetime volumes. The resulting QSEIs turn out to depend on the state of interest. We analyse the state-dependence of these bounds in Minkowski spacetime for thermal (KMS) states, and show that the lower bounds grow more slowly in magnitude  than the EED itself as the temperature increases. The lower bounds are therefore of lower energetic order than the EED, and qualify as nontrivial state-dependent QEIs.
\end{abstract}

\maketitle

\section{Introduction}
\label{sec:introduction}
 
Quantum energy inequalities (QEIs) were introduced by Ford \cite{Ford:1978qya} 40 years ago as an explanation of why macroscopic violations of the second law of thermodynamics do not occur in quantum field theories. They provide restrictions on the possible magnitude and duration of any negative energy densities or fluxes within a quantum field theory. For a recent introduction to QEIs and summary of known results see Ref.~\cite{Fewster:2012yh} and \cite{Fewster2017QEIs}.

QEIs have been used extensively to constrain exotic spacetimes such as ones allowing superluminal travel, traversable wormholes and closed timelike curves \cite{Ford:1995wg, Pfenning:1997wh, Fewster:2005gp}. In this paper, our main interest will be in whether QEI restrictions are sufficient to prove singularity theorems for matter sources described by quantum fields.
The classical singularity theorems of Hawking and Penrose \cite{Penrose:1964wq, Hawking:1966sx} use pointwise energy conditions, which are easily violated by quantum fields. In particular, Hawking's theorem uses the strong energy condition (SEC), which requires that the effective energy density (EED)~\cite{Brown:2018hym} (cf.,~also \cite{Pirani:2009})
\be\label{eq:EED}
\rho_U := T_{\mu \nu} U^\mu U^\nu -\frac{T}{n-2}  
\ee
is non-negative. Here, $U^\mu$ is a timelike vector representing the observer's velocity, $T$ is the trace of the stress-energy tensor, and $n$ the number of spacetime dimensions.

Violations of the SEC do not necessarily mean that the conclusions of the singularity theorems no longer hold and there has been some progress in proving suitably adapted singularity results under weaker hypotheses on the EED or energy density~\cite{Tipler:1978zz,chicone1980line,Galloway:1981, Borde:1987qr, Roman:1988vv,Wald:1991xn,Borde:1994ai}. In Ref.~\cite{Fewster:2010gm} it was shown that lower bounds on local weighted averages of the EED modelled on QEIs are sufficient to derive singularity theorems of Hawking type (that is, establishing timelike geodesic incompleteness) even if the EED is not everywhere positive or has a negative long-term average. In our recent work \cite{Brown:2018hym}, we established bounds on the EED of the classical nonminimally coupled scalar field and deduced, by the methods of~\cite{Fewster:2010gm}, a Hawking-type singularity theorem for the Einstein--Klein--Gordon theory. Similar methods have been applied to prove an area theorem under weakened hypotheses~\cite{Lesourd2018}.  

Despite the progress made in proving singularity theorems with weakened energy conditions, there has not been yet a singularity theorem for matter described by a quantum field theory (QFT) based on QEIs. In the case of Hawking-type results, the first necessary step is to establish a \emph{quantum strong energy inequality} (QSEI) that provides bounds on the renormalized EED. (We have chosen the name to be reminiscent of the SEC.) In this work we establish, for the first time, various QSEIs for nonminimally coupled scalar fields, by analogy with the analysis of Ref.~\cite{Fewster:2007ec} of (quantum) energy inequalities on the energy density of the nonminimally coupled scalar field. 

In particular, we derive \textit{difference} QEIs, namely, lower bounds on the expectation value of the locally averaged quantized energy density (or similar quantities) in Hadamard state $\omega$, normal ordered relative to a reference Hadamard state $\omega_0$. Recall that Hadamard states are those whose two-point functions have a specific singularity structure, which will be described later. Schematically, difference QEIs take the form
 \begin{equation}
	\label{eqn:gendiff}
	\langle\nord{\rho}_{\omega_0}(f)\rangle_{\omega} = \langle  \rho (f)\rangle_{\omega} - \langle  \rho (f)\rangle_{\omega_0}\geq -\langle \mathfrak{Q}_{\omega_0}(f)\rangle_\omega \,,
\end{equation}
where $\langle  \rho (f)\rangle_{\omega}$ is the Hadamard-renormalized energy density (or similar) in state $\omega$, averaged against $f$, which is a non-negative test function on spacetime, or singularly supported along a timelike curve, and we speak of worldvolume or worldline averages accordingly.
Here, $\mathfrak{Q}_{\omega_0}(f)$ is allowed to be an unbounded operator. In contrast, \textit{absolute} QEI's are lower bounds on $\langle  \rho (f)\rangle_{\omega}$ that do not require a reference state.  

If $\mathfrak{Q}_{\omega_0}(f)$ is a multiple of the unit operator then the right-hand side of Eq.~\eqref{eqn:gendiff} does not depend on $\omega$ (though it will generally depend on the reference state $\omega_0$) and the QEI is called \emph{state-independent}. More generally, if the right-hand side depends non-trivially on $\omega$, inequality~\eqref{eqn:gendiff} is described as \emph{state-dependent QEI}. State-independent difference inequalities for the usual energy density (known as quantum weak energy inequalities (QWEIs)) have been proved in various situations: for example, they were proved for the minimally coupled scalar field in two and four dimensions for Minkowski spacetime \cite{Ford:1994bj}, in static spacetimes \cite{Pfenning:1997rg} and, for all Hadamard states in spacetimes with general curvature in Ref.~\cite{Fewster:1999gj}. Meanwhile, state-independent absolute bounds have been established, again for minimal coupling, in two-dimensions for flat \cite{Flanagan:1997gn} and curved spacetimes \cite{Flanagan:2002bd} and subsequently for four-dimensional curved spacetimes \cite{Fewster:2007rh, Kontou:2014tha}. We refer the reader to~\cite{Fewster2017QEIs} for more references, including results on Dirac, Maxwell and Proca fields and also some results on interacting QFTs.

On the other hand, it is known that the nonminimally coupled scalar field cannot obey a state independent QWEI, as can be seen by explicit examples \cite{Fewster:2007ec}; the same argument also applies to QSEIs. However, as shown in Ref.~\cite{Fewster:2006iy}, nonminimally coupled fields obey state dependent QWEIs of both absolute and difference types. Here, we show that they also obey state dependent difference QSEIs. While it is in principle possible to establish an absolute QSEI, this is the objective of a future work.
 
This paper is organized as follows. In Sec.~\ref{sec:quant} we discuss the quantization procedure we use, giving careful attention to the relation between quantizations based on equivalent classical expressions. In Sec.~\ref{sec:SQEI} we show that the EED obeys a QEI of the form of \eqref{eqn:gendiff} and explicitly derive bounds for worldline and worldvolume averages. In \sec\ref{sec:mink-space} we study the simplified form taken by the QSEI bounds in flat spacetimes. In \sec\ref{sec:KMS}, through explicit calculations for the family of KMS states, we are able to show that, despite being state-dependent, our QSEI bounds are non-trivial in the sense that the lower bound
is of lower energetic order than the energy density itself. Finally, we conclude in \sec\ref{sec:conclusions} with a discussion on how our results inform the Hawking singularity theorem.

We employ $[-,-,-]$ conventions in the Misner, Thorne and Wheeler classification \cite{MTW}. That is, the metric signature is $(+,-,-, \dots)$, the Riemann tensor is defined as $R^{\phantom{\lambda \eta\nu}\mu}_{\lambda \eta\nu}v^\nu=(\nabla_\lambda \nabla_\eta-\nabla_\eta \nabla_\lambda)v^\mu$, and the Einstein equation is $G_{\mu \nu}=-8\pi T_{\mu \nu}$. The d'Alembertian is written $\Box_g = g^{\mu\nu}\nabla_\mu\nabla_\nu$ and work in $n$ spacetime dimensions unless otherwise stated. We adopt units in which $G=c=1$.

\section{Quantization}
\label{sec:quant}

\subsection{Quantization of the real scalar field and Hadamard states}

Throughout the paper, we assume that the spacetime is a smooth 
$n$-dimensional Lorentzian manifold $(M,g)$ that is globally hyperbolic, 
i.e., there are no closed causal curves and the 
intersection $J^+(p)\cap J^-(q)$ of the causal future of $p$ with 
the causal past of $q$ is compact, for all points $p,q\in M$.  
The nonminimally-coupled scalar field is described by the Lagrangian density
\be
\label{eqn:lagrangian}
\mathcal{L}=\frac{\sqrt{-g}}{2} [(\nabla^\mu \phi )\nabla_\mu\phi - (m^2+\xi R)\phi^2 ] \,,
\ee
and obeys the field equation
\begin{equation}
P_\xi\phi = 0, \qquad P_\xi:=\Box_g+m^2+\xi R\,.
\end{equation}
It will be quantized using the algebraic approach, a thorough review of which is to be found in \cite{KhavkineMoretti-aqft}. Thus, quantization proceeds by the introduction of a 
unital *-algebra $\mathscr{A}(M)$ on our manifold $M$, so that self-adjoint elements of $\mathscr{A}(M)$ are observables of the theory.
The algebra is generated by elements $\Phi(f)$, where $f \in \mathscr{D}(M)$, 
the space of complex-valued, compactly-supported, 
smooth functions on M, also denoted $C^{\infty}_0(M)$. The assumption that $(M,g)$ is globally hyperbolic entails the existence of an antisymmetric bi-distribution $E_\xi(x,y)$ which is the difference of the advanced and retarded Green functions for $P_\xi$. The objects $\Phi(f)$ represent smeared quantum fields and are required to obey the following relations:
\begin{itemize}
\item
\textbf{Linearity} \\
The map $f\rightarrow\Phi(f)$ is complex-linear,
\item
\textbf{Hermiticity} \\
$\Phi(f)^* = \Phi(\overline{f})$ \qquad $\forall f \in C^{\infty}_0(M)$,
\item
\textbf{Field Equation} \\
$\Phi(P_\xi f) = 0$  \qquad $\forall f \in C^{\infty}_0(M)$,
\item
\textbf{Canonical Commutation Relations} \\
$\left[\Phi(f), \Phi(h)\right] = iE_\xi (f,h)\mathbb{1}$  \qquad $\forall f,h \in C^{\infty}_0(M)$.
\end{itemize} 
A state of the theory is a linear functional $\omega:\mathscr{A}(M)\to  \mathbb{C}$ with the interpretation that $\omega(A)$ is the expectation value of $A\in\mathscr{A}(M)$ in state $\omega$. Of particular interest is the associated two-point function $W: \mathscr{D}(M)\times \mathscr{D}(M) \rightarrow \mathbb{C}$,
\begin{equation}
W(f,h) = \omega(\Phi(f)\Phi(h))\,.
\end{equation} 
However, the definition of a state includes many that have unphysical properties.
	Moreover, there is no single distinguished state associated to each spacetime that can  act as a generalization of the Minkowski vacuum state. Therefore, what is needed is a class of physically well-behaved states in each spacetime -- a standard choice being the class of Hadamard states. See~\cite{Fewster_artofstate:2018} for a recent review of these issues. 

Hadamard states were originally defined in terms of a short-distance series expansion~\cite{KayWald:1991}, but can also be described as those whose two-point functions are distributions with a given singularity structure specified by its wave-front set, as was first realised by Radzikowski~\cite{Radzikowski:1996}.

As we now briefly recall, the wave-front set $\WF(u)$ of a distribution $u$ on a smooth $m$-dimensional manifold $X$ is a subset of the cotangent bundle $T^*X$ which encodes both positional and directional information concerning the singularities of $u$. A coordinate-independent definition may be given as follows (see~\cite{Duistermaat_FIO,Brouder_etal:2014}) using the convention that, for any smooth real-valued function $\psi$ on $X\times\mathbb{R}^m$ and any fixed $a\in \mathbb{R}^m$, $\psi_a$ denotes the smooth function $\psi_a(x)=\psi(x,a)$ on $X$. 
\begin{definition}
	A point $(x,k)\in T^*X$, with $k\neq 0$, is said to be \emph{regular} for $u$ if for each smooth real-valued function $\psi$ on $X\times \mathbb{R}^m$ with $d\psi_0|_{x}=k$, there are open neighbourhoods $U$ of $x\in X$ and $A$ of $0\in \mathbb{R}^m$ so that 
	\begin{equation}
	\sup_{(\lambda,a)\in(0,\infty)\times A} \lambda^N |u(e^{i\lambda \psi_a}\phi)|<\infty\qquad\text{for all $N\in\mathbb{N}$ and $\phi\in C_0^\infty(U)$.}
	\end{equation} 
	(That is, $u(e^{i\lambda \psi_a}\phi)$ decays faster than any inverse power of $\lambda$ as $\lambda\to+\infty$, uniformly in $a\in A$.)    
	The \emph{wavefront set} $\WF(u)$ is the set of all $(x,k)\in T^*X$ with $k\neq 0$ that are not regular for $u$. 
\end{definition}
We remark that, if $(x,k)$ is regular for $u$, then one may replace $\mathbb{R}^m$ by any $\mathbb{R}^p$ and the same decay properties will continue to hold.

The Hadamard condition may now be stated as follows.
\begin{definition}
	A state $\omega$ is \emph{Hadamard} if its two-point function $W$ is a distribution on $M\times M$ whose wavefront set obeys
	\begin{equation}\label{eq:uSC}
	\WF(W)\subset \mathcal{N}^+\times\mathcal{N}^-,
	\end{equation}
	where $\mathcal{N}^\pm\subset T^*M$ is the set of positive/negative-frequency null covectors and we identify $T^*(M\times M)$ with $T^*M\times T^*M$. 
	[The sign of the frequency of a null covector is given by the sign of its contraction with any future-directed timelike vector].
\end{definition}

A remarkable fact is that the Hadamard condition~\eqref{eq:uSC}, together with the algebraic relations in $\mathscr{A}(M)$, fixes the two-point function up to smooth terms (see \cite{KhavkineMoretti-aqft,Fewster_artofstate:2018} for reviews and original references).
In particular, the difference of any two Hadamard two-point functions is smooth.  

For some purposes, we will consider states that are both Hadamard and \emph{quasifree}, meaning that all odd $n$-point functions vanish and all even $n$-point functions can be expanded as sums of products of the two-point function according to Wick's theorem, giving in particular  
\begin{equation}
\omega(\Phi(f)^n) = i^{-n}\left.\frac{d^n}{d\lambda^n}\exp \left(-\frac{\lambda}{2} W(f,f)\right)\right|_{\lambda=0}, \qquad f\in \mathscr{D}(M) \,.
\end{equation}
Each quasifree state can be represented by the vacuum vector in a suitable
Fock space representation of the algebra.

\subsection{Quantization of Wick polynomials and the stress tensor}

The algebra $\mathscr{A}(M)$ does not contain elements that correspond to smeared local Wick polynomials of degree $2$ and above; in particular, it does not contain smearings of the stress-energy tensor. These objects appear as elements
of an extended algebra $\mathscr{W}(M)$ whose construction is described in~\cite{Hollands:2001nf} and which contains $\mathscr{A}(M)$ as a subalgebra. We sketch only the parts of the discussion needed here,
suppressing many points of detail and slightly changing conventions and notation.

To start, let $\omega$ be a quasifree Hadamard state with two-point function $W$.
Then the algebra $\mathscr{A}(M)$ contains elements of the form
\begin{equation}\label{eq:Wkex}
\nord{\Phi^{\otimes 2}}_\omega(f\otimes f) = \Phi(f)\Phi(f) - W(f,f)\mathbb{1}.
\end{equation}
for any test function $f$, with the property that $\langle\nord{\Phi^{\otimes 2}}_\omega(f\otimes f)\rangle_\omega=0$. More generally, the extended algebra $\mathscr{W}(M)$ contains elements $\nord{\Phi^{\otimes 2}}_\omega(t)$, where $t$ is any symmetric, compactly supported distribution on $M\times M$ whose wave-front set does not contain any points $(x,k;x',k')\in T^*(M\times M)$ in which $k$ and $k'$ are causal covectors that both have positive frequency or both have negative frequency. 
For example, the elements given in Eq.~\eqref{eq:Wkex} correspond to the case in which $f\otimes f$ is regarded as acting on smooth functions $S$ on $M\times M$ by 
\begin{equation}
(f\otimes f)(S) = \int_{M\times M} \dV_x\, \dV_y\, f(x)f(y) S(x,y) .
\end{equation}
Any quasifree Hadamard state $\omega$ extends to $\mathscr{W}(M)$ so that
$\langle\nord{\Phi^{\otimes 2}}_\omega(t)\rangle_\omega=0$ for all $t$ of the type just described.

More generally, distributions that involve (derivatives of) $\delta$-functions may be used to define Wick polynomials. It will be enough for our purposes to introduce quadratic Wick polynomials in the field and its derivatives. Let $f^{\mu_1\cdots \mu_r \nu_1\cdots \nu_s}$ be a smooth compactly supported tensor field and define a compactly
supported distribution $T^{r,s}[f]$ by
\begin{equation}\label{eq:Trs_def}
T^{(r,s)}[f](S) = \int_{M } \dV\, f^{\mu_1\cdots \mu_r \nu_1\cdots \nu_s} \coin{(\nabla^{(r)}\otimes \nabla^{(s )}) S^{\mathrm{sym}}}_{\mu_1\cdots \mu_r \nu_1\cdots \nu_s}.
\end{equation} 
Here, we have written
$S^{\mathrm{sym}}(x,y)=\frac{1}{2} (S(x,y)+S(y,x))$ for the symmetric part of $S\in C^\infty(M\times M)$, while $\nabla^{(r)}$ is a symmetrised $r$-th order covariant derivative and the double square brackets $\coin{\cdot}$ in the integrand denote a coincidence limit. Then we obtain a smeared Wick polynomial
\begin{equation}
\nord{\nabla^{(r)}\Phi \nabla^{(s)}\Phi}_\omega (f):= \nord{\Phi^{\otimes 2}}_\omega(T^{r,s}[f]),
\end{equation}
which depends on the reference state $\omega$. As a matter of fact, one has 
\begin{equation}
\nord{\nabla^{(r)}\Phi \nabla^{(s)}\Phi}_{\omega} (f)= \nord{\nabla^{(r)}\Phi \nabla^{(s)}\Phi}_{\omega'} (f) + 
T^{r,s}[f](W'- W)\mathbb{1},
\end{equation}
if $\omega'$ is another quasifree Hadamard state with two-point function $W'$. 
Taking expectations in the state $\omega'$ yields
\begin{equation}
\langle \nord{\nabla^{(r)}\Phi \nabla^{(s)}\Phi}_{\omega} (f)\rangle_{\omega'} = 
T^{r,s}[f](W'-W),
\end{equation}
which reproduces the usual point-splitting regularisation for normal ordering with respect to $\omega$. 

Standard results concerning coincidence limits may be used to manipulate expressions of the form $T^{r,s}[f](S)$. For example, the identity
\begin{equation}
Y^\mu Z^\nu  \coin{C}_{\mu\nu} = \coin{Y^\mu Z^{\nu'}C_{\mu\nu'}}
\end{equation}
satisfied by continuous vector fields $Y^\mu$, $Z^\mu$ and bi-covector field $C_{\mu\nu'}(x,x')$ 
implies that
\begin{equation}\label{eq:T11id}
T^{1,1}[(Y\otimes Z) f](S) = T^{0,0}[f]((\nabla_Y\otimes \nabla_Z)S).
\end{equation}
Similarly, if $C$ is now a bi-tensor field $C_{\mu'\nu'}(x,x')$ [i.e., scalar type with respect to $x$, and second rank covariant with respect to $x'$], the identity
\begin{equation}
Y^\mu Z^\nu  \coin{C}_{\mu\nu} = \coin{Y^{\mu'} Z^{\nu'}C_{\mu'\nu'}}
\end{equation}
implies
\begin{equation}\label{eq:T02id}
T^{0,2}[(Y\otimes Z) f](S) = T^{0,0}[f]((1\otimes Y^{\mu'}Z^{\nu'}\nabla_{(\mu'}\nabla_{\nu')}) S). 
\end{equation}  

The dependence of the above normal-ordered expressions on $\omega$ is unsatisfactory, because of the lack of a canonical choice of a Hadamard state in a general curved spacetime. What is needed, therefore, is
a prescription for finding algebra elements that qualify as local and covariant Wick powers. This might be done in various ways, reflecting finite renormalisation freedoms. Hollands and Wald~\cite{Hollands:2001nf,Hollands:2004yh} set out a list of axioms (labelled T1--T11) that should be 
obeyed by any reasonable scheme and which moreover encompasses time ordered expressions. Among their
requirements is a form of Leibniz' rule (T10) (related to the `action Ward identity'~\cite{DutschFredenhagen:2004}) which in our case implies, for example, that
\begin{equation}\label{eq:egLeibniz0}
(\nabla_\mu\Phi^2)(f^\mu) = 2(\Phi\nabla_\mu\Phi)(f^\mu)
\end{equation}
and 
\begin{equation}\label{eq:egLeibniz}
\frac{1}{2}(\nabla_\mu \nabla_\nu (\Phi^2))(f^{\mu\nu}) = (\nabla_\mu\Phi\nabla_\nu\Phi)(f^{\mu\nu}) + (\Phi\nabla_{(\mu}\nabla_{\nu)}\Phi)(f^{\mu\nu}),
\end{equation}
where in each case the left-hand side is understood distributionally, i.e., 
\begin{equation}
	\label{eqn:Leibdist}
(\nabla_\mu\Phi^2)(f^\mu) = -\Phi^2(\nabla_\mu f^\mu), \qquad
(\nabla_\mu \nabla_\nu (\Phi^2))(f^{\mu\nu})  = (\Phi^2)(\nabla_\nu\nabla_\mu f^{\mu\nu}) .
\end{equation} 

While the Leibniz rule must hold for all Wick ordering prescriptions, it is not generally true that
the field equation can be imposed inside Wick ordered expressions. In fact, Hollands and Wald showed~\cite{Hollands:2004yh} that
one cannot consistently impose both $\Phi P_\xi \Phi = 0$ and $(\nabla_a\Phi) P_\xi\Phi = 0$ in 
$n=4$ spacetime dimensions, and that the latter cannot be imposed in $n=2$ dimensions by any prescription obeying their axioms. However, these fields are at least given by local curvature tensors (of an appropriate rank and engineering dimension) multiplied by the identity element. For example, one has 
\begin{equation}\label{eq:Q}
(\Phi P_\xi \Phi)(f) = \int_M \dV\, f Q\, \mathbb{1}, 
\end{equation}
where $Q$ is a scalar quantity, locally and covariantly constructed from the metric (including curvature tensors and covariant derivatives thereof) and the parameters $m^2$ and $\xi$, and with overall engineering dimension of $n-2$ powers of inverse length. The dependence on $m^2$ and $\xi$ is restricted in certain ways.

Turning to the stress-energy tensor, the classical expression obtained
by varying the action derived from~\eqref{eqn:lagrangian} with respect to the metric 
is
\be
\label{eqn:tmunu}
T_{\mu \nu}=(\nabla_\mu \phi)(\nabla_\nu \phi)+\frac{1}{2} g_{\mu \nu} (m^2 \phi^2-(\nabla \phi)^2)+\xi(g_{\mu \nu} \Box_g-\nabla_\mu \nabla_\nu-G_{\mu \nu}) \phi^2 \,,
\ee
where $G_{\mu \nu}$ is the Einstein tensor. The stress-energy tensor can be 
expressed in terms of $\phi^2$ and $\phi\nabla_\mu\nabla_\nu\phi$, using Leibniz' rule but without using the field equation.
Therefore any Wick ordering prescription gives a quantized stress-energy tensor in terms of the 
Wick ordered expressions 
$(\Phi^2)$ and $(\Phi\nabla_{(\mu}\nabla_{\nu)}\Phi)$. 
Classically, Leibniz' rule also gives $\nabla^\mu T_{\mu\nu}= (\nabla_\nu\phi) P_\xi\phi$, so the conservation of the quantized stress-energy tensor
requires that a prescription with $((\nabla_\mu\Phi) P_\xi\Phi) = 0$ is adopted (and therefore, in this approach, the stress-energy tensor cannot be conserved in $n=2$ dimensions). It turns out that stress-energy tensor conservation (if $n>2$) is a consequence of the requirement (T11) imposed by Hollands and Wald~\cite{Hollands:2004yh}, which is inspired by a `principle of perturbative agreement' and furthermore guarantees conservation of the stress-energy tensor in perturbatively constructed interacting models.  

A prescription meeting the requirements discussed so far may be given as follows. First, let $H$ be a local, symmetric Hadamard parametrix, defined near the diagonal in $M\times M$. (For the definition and properties of the Hadamard parametrix, see for example~\cite{KhavkineMoretti-aqft} and references given there.) Then the prescription
\begin{equation}
(\nabla^{(r)}\Phi\nabla^{(s)}\Phi)_H (f) = \nord{\nabla^{(r)}\Phi\nabla^{(s)}\Phi}_\omega (f)  + 
T^{r,s}[f](W^\textrm{sym} - H)\mathbb{1}
\end{equation}
satisfies requirements (T1--10) but not (T11).
Because the distribution $T^{r,s}[f]$ is supported on the diagonal in $M\times M$, the fact that $H$ is only defined near the diagonal is harmless. In this prescription, it is known~\cite{Moretti:2001qh} that 
\begin{equation}
((\nabla_\mu\Phi) P_\xi\Phi)_H = \frac{n}{2(n+2)}\nabla_\mu Q \mathbb{1}
\end{equation}	
if $Q$ is defined as in~\eqref{eq:Q} for $(\Phi P_\xi\Phi)_H$.  
Therefore, adopting a prescription in which $(\Phi^2)=(\Phi^2)_H$ but 
\begin{equation}
(\Phi\nabla_{(\mu}\nabla_{\nu)}\Phi)   = (\Phi\nabla_{(\mu}\nabla_{\nu)}\Phi)_H - \frac{n}{n^2-4} g_{\mu\nu} Q\mathbb{1}
\end{equation}
will result in a stress-energy tensor that is automatically conserved.\footnote{Here, we have corrected the corresponding expression in \cite{Hollands:2004yh} which contains some errors.} Of course, the definition of
$(\nabla\Phi\nabla\Phi)$ must now differ from $(\nabla\Phi\nabla\Phi)_H$ in order to protect the Leibniz rule
\eqref{eq:egLeibniz},
\begin{equation}
(\nabla_\mu\Phi\nabla_\nu\Phi) = (\nabla_\mu\Phi\nabla_\nu\Phi)_H + \frac{n}{n^2-4} g_{\mu\nu} Q\mathbb{1}\,.
\end{equation}
Likewise, the prescription for a derivative $(\Phi\nabla\Phi)$ is fixed
by the Leibniz rule
\begin{equation}
(\Phi\nabla\Phi) = \frac{1}{2}\nabla (\Phi^2).  
\end{equation}

 Hollands and Wald showed that one can inductively modify the Hadamard prescription so that all Wick and time ordered expressions are consistent with all their requirements~\cite{Hollands:2004yh}. There remain further finite renormalisation freedoms, for example, in selecting a length scale that is needed in the construction of $H$. These result in the freedom to add multiples of $\mathbb{1}$ to $T_{\mu\nu}$, given by conserved local curvature terms.

This construction supplants the older viewpoint, see e.g.~\cite{Wald:1995yp}, in which one renormalises the
stress-energy tensor directly using point-splitting and a Hadamard subtraction and then makes an \emph{ad hoc}
modification to fix the failure of conservation. Instead, conservation follows from wider requirements
on the time ordering prescription. If one is only interested in defining the stress-energy tensor, one can proceed alternatively by modifying the classical expression for $T_{\mu\nu}$, adding a term proportional to $g_{\mu\nu} \phi P_\xi\phi$ that vanishes on shell, as shown by Moretti~\cite{Moretti:2001qh}. This gives identical results to the more general Hollands--Wald prescription~\cite{Hollands:2004yh} in $n=4$ dimensions, but not in $n=2$. 

When taking differences of expectation values, multiples
of the unit cancel. Therefore, if $\omega$ and $\omega'$ are quasifree Hadamard states, 
the difference in expectation values of any quadratic Wick expression is given
by the point-splitting result
\begin{align}
\langle (\nabla^{(r)}\Phi\nabla^{(s)}\Phi)(f)\rangle_{\omega'} - 
\langle (\nabla^{(r)}\Phi\nabla^{(s)}\Phi)(f)\rangle_{\omega} &= 
\langle \nord{(\nabla^{(r)}\Phi\nabla^{(s)}\Phi)}_{\omega}(f)\rangle_{\omega'} \notag\\
&= T^{r,s}[f](W'-W).
\end{align}
Furthermore, because the difference $W'-W$ is a smooth
bisolution to the operator $P_\xi$, one can use the field equation in the sense that
\begin{equation}
\langle (\nabla^{(r)}\Phi P_\xi\Phi)(f)\rangle_{\omega'} - 
\langle (\nabla^{(r)}\Phi P_\xi\Phi)(f)\rangle_{\omega} = 0.
\end{equation}
That is, while the quadratic Wick ordered expressions obey Leibniz' rule, but not generally the field equation, the differences in their expectation values obey both. 
This gives us the freedom to quantise classical expressions in the most convenient
fashion for proving quantum energy inequalities.

\subsection{Quantization of the effective energy density}

The main observable of interest will be the EED, classically defined by~\eqref{eq:EED}, where $U^\mu$ is the velocity field of a family of observers. As a quantum field, $\rho_U$ may be defined by 
\begin{equation}
\rho_U(f) = T_{\mu\nu}\left(\left(U^\mu U^\nu - \frac{g^{\mu\nu}}{n-2}\right)f\right),
\end{equation}
where $T_{\mu\nu}$ 
is the quantized stress tensor constructed as described above. Using the Leibniz rule, $\rho_U(f)$ may be written in various ways, which will be useful for different purposes. If the second derivatives of $\Phi^2$ arising from Eq.~\eqref{eqn:tmunu} are expanded using Eq.~\eqref{eq:egLeibniz}, the expression
 \begin{align}
	\label{eqn:EED1}
\rho_U(f) &= (\nabla_\mu\Phi\nabla_\nu\Phi)\left(\left( (1-2\xi)U^\mu U^\nu -\frac{2\xi g^{\mu\nu}}{n-2}\right)f \right) - 2\xi(\Phi\nabla_{(\mu}\nabla_{\nu)}\Phi)(U^\mu U^\nu f) \notag\\
&\qquad +\Phi^2\left(\left( \xi \mathcal{R}_\xi -\frac{1-2\xi}{n-2} m^2  \right)f\right) - \frac{2\xi}{n-2}(\Phi P_\xi\Phi)(f),
\end{align}
is obtained, where we have also used the definition of $P_\xi$ (but not the
field equation) to absorb the $\Phi\Box_g\Phi$ term, and defined
\begin{equation}\label{eq:R}
\mathcal{R}_\xi=\frac{2\xi}{n-2} R  -R_{\mu \nu} U^\mu U^\nu\,.
\end{equation}
Eq.~\eqref{eqn:EED1} corresponds to a classical expression used in~\cite{Brown:2018hym}. On the other hand, if Eq.~\eqref{eq:egLeibniz} is applied again, along with Eq.~\eqref{eqn:Leibdist}, to rewrite Eq.~\eqref{eqn:EED1} in terms of
the Wick polynomials $\Phi^2$, $\nabla_\mu\Phi\nabla_\nu\Phi$ and $\Phi P_\xi\Phi$ alone, we find
\begin{align}
\label{eqn:EED1b}
\rho_U(f) &= (\nabla_\mu\Phi\nabla_\nu\Phi)\left(\left( U^\mu U^\nu -\frac{2\xi g^{\mu\nu}}{n-2}\right)f \right)\notag\\
&\qquad  -\Phi^2\left(\xi \nabla_\nu\nabla_\mu (U^\mu U^\nu f)+\left(\frac{1-2\xi}{n-2} m^2 - \xi \mathcal{R}_\xi\right)f\right) - \frac{2\xi}{n-2}(\Phi P_\xi\Phi)(f)\,.
\end{align} 

Alternatively, the mass term in Eq.~\eqref{eqn:EED1b} may be traded for additional terms involving $P_\xi$, $\Box_g$ and the Ricci scalar, giving
\begin{align}
\label{eqn:EED2}
\rho_U(f) &= (\nabla_\mu\Phi\nabla_\nu\Phi)\left(\left( U^\mu U^\nu -\frac{g^{\mu\nu}}{n-2}\right)f \right) + \frac{1-2\xi}{2(n-2)} \Phi^2(\Box_g f)   \notag\\
& \qquad -\xi \Phi^2\left( \nabla_\nu\nabla_\mu (U^\mu U^\nu f)-\mathcal{R}_{1/2} f\right) - \frac{1}{n-2}(\Phi P_\xi\Phi)(f),
\end{align}
 in which the mass parameter appears only in the last term. The curvature term $\mathcal{R}_{1/2}$ is just $\mathcal{R}_\xi$ with $\xi=1/2$. The three expressions for $\rho_U$ are all equivalent, but have different advantages as starting points for quantum energy inequalities. 

We will be interested in expectation values of the quantized EED in state $\omega'$, normal ordered relative to a reference Hadamard state $\omega$,
\begin{equation}
	\label{eqn:expecteed}
\langle\nord{\rho_U}_{\omega}(f)\rangle_{\omega'} = \langle  \rho_U (f)\rangle_{\omega'} - \langle  \rho_U (f)\rangle_{\omega}\,.
\end{equation}
 Each term in the above expressions~\eqref{eqn:EED1},~\eqref{eqn:EED1b} or~\eqref{eqn:EED2} may then be written in terms of distributions $T^{r,s}[\cdot]$ acting on the difference of the two-point functions $S=W'-W$. By further manipulation, they may all be expressed in terms of $T^{0,0}[f]$ acting on suitable derivatives of $S$. For instance, if $V^\mu$ is any smooth vector field,
\begin{equation}\label{eq:Wick_id1}
 \langle \nord{\nabla_\mu\Phi\nabla_\nu\Phi}_{\omega}(V^\mu V^\nu f)\rangle_{\omega'}
 = T^{1,1}[(V\otimes V) f](S) = T^{0,0}[f]((\nabla_V\otimes\nabla_V)S)\,,
\end{equation}
where we have used the identity~\eqref{eq:T11id}. Similarly, if $e^\mu_a$ ($a=0,\ldots,n-1$)
is an $n$-bein defined on the support of $f$ with $e_0$ timelike, we also find 
\begin{equation}\label{eq:Wick_id2}
\langle \nord{\nabla_\mu\Phi\nabla_\nu\Phi}_{\omega}(g^{\mu\nu} f)\rangle_{\omega'} = 
T^{0,0}[f]((\nabla_{e_0}\otimes\nabla_{e_0})S) - \sum_{a=1}^{n-1} T^{0,0}[f]((\nabla_{e_a}\otimes\nabla_{e_a})S) .
\end{equation}
Finally, the identity~\eqref{eq:T02id} yields 
\begin{align}\label{eq:Wick_id3}
\langle (\nord{\Phi\nabla_{(\mu}\nabla_{\nu)}\Phi}_{\omega})(U^\mu U^\nu f)\rangle_{\omega'} &= T^{0,2}[(U\otimes U) f](S) =
T^{0,0}[f]((1\otimes_{\mathfrak{s}} U^\mu U^\nu \nabla_\mu \nabla_\nu) S) \notag\\
&=
T^{0,0}[f]((1\otimes_{\mathfrak{s}} U^\mu U^\nu \nabla_\mu \nabla_\nu) S) \notag\\
&=
T^{0,0}[f]((1\otimes_{\mathfrak{s}} (\nabla_U^2 - \nabla_{\nabla_U U})) S),
\end{align}
where $\otimes_{\mathfrak{s}}$ is the symmetrised tensor product $P \otimes_{\mathfrak{s}} P'=[ (P \otimes P')+(P' \otimes P)]/2$. Note that expectation values of $\nord{\Phi P_\xi \Phi}$ vanish. 
In this way, the expectation values of normal ordered quantities may be reduced to coincidence limits of certain differential operators acting on the difference of two-point functions. For instance,
\begin{equation}
\langle V^\mu V^\nu\nord{\nabla_\mu\Phi\nabla_\nu\Phi}_{\omega} \rangle_{\omega'}(x)
= \coin{(\nabla_V\otimes\nabla_V)(W'-W)}(x).
\end{equation}
 This is just as in the traditional viewpoint of renormalisation by point-splitting~\cite{Wald:1995yp}, but with the advantage that there is a systematic framework from which suitable differential operators in question may be derived, rather than simply being asserted as  Ans\"atze. Although we used vielbeins to treat terms of the form $\nord{(\nabla\Phi)^2}$, it would be equally valid to use parallel propagators, leading ultimately to the same expectation values in the end.

\section{Quantum strong energy inequalities}
\label{sec:SQEI} 

Having described in detail how the EED may be quantised, we now turn to the derivation of QSEIs for averaging along timelike worldlines or spacetime volumes.

\subsection{Worldline}
\label{sub:qline} 
Let $\gamma$ be a smooth timelike curve parametrized by proper time $\tau$. 
Choose any smooth $n$-bein $e_a$ ($a=0,\ldots,n-1)$ on a tubular neighbourhood $\mathcal{T}$ of $\gamma$, so that $U^\mu=e_0^\mu$ is everywhere timelike and agrees with $\dot{\gamma}^\mu$ on $\gamma$.
Fix a Hadamard reference state $\omega_0$ with $2$-point function $W_0$ and, for brevity, denote all
quantities normal-ordered relative to $\omega_0$ by $\nord{X}$, rather than $\nord{X}_{\omega_0}$.
Using the procedure described in the previous subsection, the expectation values of the effective energy density $\nord{\rho_U}$ in Hadamard state $\omega$ can be written in terms of the coincidence limits acting on $\nord{W}=W-W_0$. 
Eq.~(\ref{eqn:EED1}), together with identities~\eqref{eq:Wick_id1}, \eqref{eq:Wick_id2} and \eqref{eq:Wick_id3}), gives
\begin{equation}\label{eqn:rhoquantum}
\langle \nord{\rho_{U}}\rangle_\omega =  \coin{ \hat{\rho}_1 \nord{W} }  +\coin{ \hat{\rho}_2 \nord{W}}  + \left(\xi\mathcal{R}_\xi-\frac{1-2\xi}{n-2} m^2 \right) \coin{\nord{W}}
\end{equation}
along $\gamma$, where the operators $\hat{\rho}_i$ are given by 
\blea
\label{eqn:lineoper}
\hat{\rho}_1&=&\left( 1-2\xi \frac{n-1}{n-2} \right) (\nabla_{U} \otimes \nabla_{U} )+\frac{2\xi}{n-2} \sum_{a=1}^{n-1} (\nabla_{e_a} \otimes \nabla_{e_a})  \,, \\
\hat{\rho}_2&=& -2\xi (\mathbb{1} \otimes_{\mathfrak{s}} U^\mu U^\nu \nabla_\mu \nabla_\nu )  
\,.
\elea 
We have used the fact that the field equation holds for normal-ordered expressions. 
Note that $\mathcal{R}_\xi$ vanishes for Ricci-flat spacetimes.

Our aim is to establish QEI lower bounds on the averaged EED along $\gamma$,
\be 
\langle \nord{\rho_U} \circ \gamma \rangle_\omega (f^2)=\int d\tau f^2 (\tau) \langle  \nord{\rho_U}\rangle_\omega   (\gamma(\tau)) \,,
\ee 
where $f \in \mathcal{D} (\mathbb{R}, \mathbb{R})$ is a real valued test function. 
The contributions arising from the three terms in~\eqref{eqn:rhoquantum} will be handled in differing ways. Note first that all the terms in $\hat{\rho}_1$ take the form
$Q\otimes Q$ for some partial differential operator $Q$ with real coefficients, 
provided $\xi \in [0,2\xi_c]$, where
\begin{equation}
\xi_c=\frac{n-2}{4(n-1)}
\end{equation}
is the value of $\xi$ corresponding to conformal coupling.
The contribution of the terms deriving from $\hat{\rho}_1$ to the averaged EED can be bounded from below, uniformly in $\omega$, using the methods of~\cite{Fewster:1999gj} (see also Lemma~\ref{lem:Qline} below). A key point here is that operators $Q\otimes Q$ map any positive type bi-distribution to another positive type bi-distribution. By contrast, the the mass term is negative definite for $\xi<1/2$, while the geometric term $\mathcal{R}_\xi$ has no
definite sign in general and for this reason cannot be bounded below by a state-independent QEI. (However, see the remarks following Theorem~\ref{the:linegen}.) This leaves $\hat{\rho}_2$, the contribution of which can be manipulated to a more convenient form using the following lemma.
\begin{lemma}
	If $F$ is a smooth function on $\mathcal{T}\times\mathcal{T}$ and $f\in C_0^\infty(\mathbb{R})$ then
	\begin{align}
	\int d\tau f(\tau)^2 \coin{(\mathbb{1} \otimes_{\mathfrak{s}} U^\mu U^\nu \nabla_\mu\nabla_\nu) F}(\gamma(\tau)) 
	&= - \int d\tau  \coin{(\partial\otimes\partial)\left((f\otimes f)\phi^* F\right)} (\tau)  \nonumber\\  &\qquad
	+\int d\tau   f'(\tau)^2 \coin{F} (\gamma(\tau)) \nonumber\\
	&\qquad -\frac{1}{2}\int d\tau \, f(\tau)^2 (\nabla_A\coin{F})(\gamma(\tau))  \,, 
	\end{align}
	where $\phi^*$ denotes a pull-back by $\phi(\tau,\tau')=(\gamma(\tau),\gamma(\tau'))$,
	$A^\mu=\nabla_U U^\mu$ is the acceleration field of $U$ and $\partial$ denotes the derivative on $\mathbb{R}$.
\end{lemma}	
\begin{proof}
	First, choose $f_\mathcal{T} \in \mathcal{D}(\mathcal{T},\mathbb{R})$ such that $f_\mathcal{T} \circ \gamma=f$. Then, slightly simplified, the identity Eq.~(38) of Ref.~\cite{Fewster:2007ec} (a consequence of Synge's rule
	$\nabla_V\coin{H}=2\coin{(\mathbb{1}\otimes_{\mathfrak{s}}\nabla_V)H}$) gives
	\begin{align} \label{eqn:identity}
	& 2f_\mathcal{T}^2\coin{(\mathbb{1} \otimes_{\mathfrak{s}} \nabla_{U}^2) F}+\nabla_{U} \coin{(\mathbb{1} \otimes_{\mathfrak{s}} (\nabla_{U} f_\mathcal{T}^2) )F} \nonumber\\
&\qquad\qquad\qquad\qquad\qquad   =-2 \coin{(\nabla_{U} \otimes \nabla_{U}) ((f_\mathcal{T} \otimes f_\mathcal{T} )F) }  +2(\nabla_{U} f_\mathcal{T})^2 \coin{  F} \nonumber\\ 
	&\qquad \qquad\qquad\qquad\qquad\qquad
	+2 \nabla_{U} \coin{(\mathbb{1} \otimes_{\mathfrak{s}} \nabla_{U})  ((f_\mathcal{T} \otimes f_\mathcal{T} )F)} \,.
	\end{align} 
	Integrating both sides along $\gamma$ and dividing by $2$, we have
	\begin{align}
	\int d\tau f(\tau)^2 \coin{(\mathbb{1} \otimes_{\mathfrak{s}} \nabla_{U}^2) F}(\gamma(\tau)) &=
	- \int d\tau  \coin{(\partial\otimes\partial)\left((f\otimes f)\phi^* F\right)} (\tau) \nonumber\\&\qquad\qquad 
	+\int d\tau ( f'(\tau))^2 \coin{\phi^* F} (\tau)  \,,
	\end{align}
	and the result follows on noting that 
	\begin{equation}
	\coin{(\mathbb{1} \otimes_{\mathfrak{s}} U^\mu U^\nu \nabla_\mu\nabla_\nu) F} = 
	\coin{(\mathbb{1} \otimes_{\mathfrak{s}} \nabla_U^2) F} - \frac{1}{2}\nabla_A \coin{F}
	\end{equation}
	using the Leibniz and Synge rules. 
\end{proof}
It follows that 
\begin{align}\label{eq:rho2expr}
\int d\tau f(\tau)^2 \coin{ \hat{\rho}_2 \nord{W}}(\gamma(\tau)) &=
2\xi \int d\tau  \coin{(\partial\otimes\partial)\left((f\otimes f)\phi^* \nord{W} \right) } (\tau) \nonumber\\&\qquad
-2\xi \int d\tau  \,  f'(\tau)^2 \langle \nord{\Phi^2} \rangle_\omega(\gamma (\tau)) \nonumber\\
&\qquad
+\xi \int d\tau f(\tau)^2 (\nabla_A \langle\nord{\Phi^2} \rangle_\omega)(\gamma (\tau))  . 
\end{align}
The first term in this expression can be estimated as a special case ($k=0$, $Q=\mathbb{1}$) of the following result, which can also be used to bound all terms arising from $\hat{\rho}_1$. (For the benefit of the reader, we note that on the right-hand side of~\eqref{eq:lem2}, the functions $\bar{f_\alpha}$ and $f_\alpha$ are substituted into a bidistribution obtained as a pull-back of $(Q\otimes Q) W_0$.)
\begin{lemma}
	\label{lem:Qline}
	Let $Q$ be a partial differential operator on $M$ with smooth real coefficients and $k\in\mathbb{N}_0$.
	Then the inequality
	\begin{align}\label{eq:lem2}
	\lefteqn{\int d\tau  \coin{(\partial^k\otimes\partial^k)\left((f\otimes f)\phi^* ((Q\otimes Q)\nord{W})\right) } (\tau) } \qquad\qquad\qquad\qquad\qquad\qquad &\notag \\ &\ge  - \int_0^\infty \frac{d\alpha}{\pi} \alpha^{2k} \left(\phi^*((Q\otimes Q) W_0)\right)(\bar{f_\alpha}, f_\alpha) > -\infty\,,
	\end{align}  
	holds for all Hadamard states $\omega$ and real-valued test functions $f\in\mathscr{D}(\mathbb{R},\mathbb{R})$,
	where $f_\alpha(\tau)=e^{i\alpha \tau} f(\tau)$.
\end{lemma}
\begin{proof}
	This is a slight generalisation of an argument first given in \cite{Fewster:1999gj}. First note that one has, for any smooth symmetric function $S$ and test function $f$ as above,
	\begin{align}
	\int d\tau  \coin{(\partial^k\otimes\partial^k)\left((f\otimes f) S\right)}(\tau)& =
	\int d\tau\,d\tau'  \delta(\tau-\tau') \partial_\tau^k\partial_{\tau'}^k\left((f\otimes f) S\right)(\tau,\tau')\nonumber \\
	&= \int_{-\infty}^\infty \frac{d\alpha}{2\pi}
	\int d\tau\,d\tau'  e^{-i\alpha(\tau-\tau')}
	\partial_\tau^k\partial_{\tau'}^k\left((f\otimes f) S\right)(\tau,\tau')\nonumber \\
	&= \int_{0}^\infty \frac{d\alpha}{\pi}
	\int d\tau\,d\tau' \, \alpha^{2k}e^{-i\alpha(\tau-\tau')}f(\tau)f(\tau')\,
	S(\tau,\tau')\,,
	\end{align}
	where in the second step we have inserted the Fourier representation of the $\delta$-function and in the last step used symmetry of $S$ and also integrated by parts $k$ times in both $\tau$ and $\tau'$. Applying this to $S=\phi^*((Q\otimes Q) \nord{W}$, 
	\begin{align}\label{eq:calc}
	\text{L.H.S. of~\eqref{eq:lem2}} &= \int_0^\infty \frac{d\alpha}{\pi}\alpha^{2k} \left(\phi^*((Q\otimes Q) \nord{W}) \right)(\bar{f_\alpha}, f_\alpha) \nonumber\\ 
	&= \int_0^\infty \frac{d\alpha}{\pi}\alpha^{2k} \left(\left(\phi^*((Q\otimes Q)  W)\right)(\bar{f_\alpha},f_\alpha)- \left(\phi^*((Q\otimes Q)  W_0)(\bar{f_\alpha},f_\alpha) \right)\right)\,,
	\end{align}
	noting that expressions of the form  $\phi^*((Q\otimes Q) W)$ are shown to exist in \cite{Fewster:1999gj}, with wave-front sets obeying
	\begin{equation}
	\WF(\phi^*((Q\otimes Q) W))\subset (\mathbb{R}\times\mathbb{R}^+)\times (\mathbb{R}\times\mathbb{R}^-)\subset T^*\mathbb{R}\times T^*\mathbb{R}.
	\end{equation}
	Together with other results proved in Ref.~\cite{Fewster:1999gj}, this shows that the two terms in the integrand in~\eqref{eq:calc} are non-negative and decay rapidly as $\alpha \to +\infty$ for any Hadamard state $\omega$ (see Theorem 2.2 of Ref.~\cite{Fewster:1999gj}). Here the microlocal properties of Hadamard states play a crucial role. Consequently, the final expression in~\eqref{eq:calc} may be written as the difference of two separately convergent nonnegative integrals. Discarding the first of these, the inequality~\eqref{eq:lem2} is proved.
\end{proof}
Applying this result to all terms arising from $\hat{\rho}_1$ and the first term in~\eqref{eq:rho2expr}, and combining with the other terms from~\eqref{eq:rho2expr} and~\eqref{eqn:rhoquantum}, we have proved the following result.

\begin{theorem} 
\label{the:linegen}
Let $W_0$ be the two-point function of a reference Hadamard state for the non-minimally coupled scalar field with coupling constant  $\xi \in \left[0,2\xi_c\right]$ and mass $m\ge 0$ defined on a globally hyperbolic spacetime $M$ with smooth metric $g$. Let $\gamma$ be a smooth timelike curve parametrised in proper time $\tau$, with velocity $U^\mu$ and acceleration $A^\mu=\nabla_U U^\mu$. Then, for all  Hadamard states $\omega$ and real-valued test functions $f \in \mathcal{D} (\mathbb{R}, \mathbb{R})$, the normal-ordered effective energy density obeys the QEI
\be \label{eqn:qboundline}
\langle \nord{\rho_U} \circ \gamma \rangle_\omega (f^2) \geq - \left( \mathfrak{Q}_1[f]  +\langle \nord{\Phi^2} \circ \gamma \rangle_\omega (\mathfrak{Q}_2[f]+\mathfrak{Q}_3[f])-\xi\langle \nabla_A \nord{\Phi^2} \circ \gamma \rangle_\omega (f^2) \right) \,,
\ee
where 
\be
\mathfrak{Q}_1[f]=\int_0^\infty \frac{d\alpha}{\pi} \left( \phi^*(\hat{\rho}_1 \, W_0)(\bar{f_\alpha},f_\alpha)+2\xi \alpha^2  \phi^* W_0  (\bar{f_\alpha},f_\alpha) \right) \,,
\ee
\be
\mathfrak{Q}_2[f](\tau)=\frac{1-2\xi}{n-2}m^2 f^2(\tau) +2\xi  ( f'(\tau))^2 \,,
\ee
and
\be
\mathfrak{Q}_3 [f](\tau)=\xi  \mathcal{R}_\xi(\gamma(\tau)) f(\tau)^2\,.
\ee
\end{theorem}
An important feature of the QEI~\eqref{eqn:qboundline}, and indeed all the QEIs that we will derive in this paper, is that the lower bound depends on the state of interest $\omega$; that is, it is a \emph{state-dependent QEI}, unlike e.g., the quantum weak energy inequality proved for the minimally coupled scalar field in~\cite{Fewster:1999gj}. Now in fact no state-independent QSEI could possibly hold (except in the massless minimally coupled case) because the classical model can violate the SEC. This makes it possible to construct single-particle quantum states relative to the Minkowski vacuum state whose averaged EED is negative for some test function. Tensoring together $N$ copies of such a state, the averaged EED scales with $N$ and so it is clear that no state-independent QSEI can be valid, even in Minkowski space. See Ref.~\cite{Fewster:2007ec}, where an analogous argument is given in detail for the energy density of the nonminimally coupled field. Nonetheless, the state-dependence of the lower bound raises concerns that will be discussed more fully in Sec.~\ref{sec:KMS}. For now we note that the only nontrivial quantum field appearing in the bound is the Wick square $\nord{\Phi^2}$ (and at most one derivative thereof), while the EED itself involves contributions involving two derivatives of $\Phi$ and squares of the derivatives of $\Phi$. This distinction will enable us to show that the QEIs we study in this paper are nontrivial. 
 
The expression for the QSEI bound simplifies in various situations: if $\gamma$ is geodesic the last term in Eq.~\eqref{eqn:qboundline} vanishes; for flat spacetimes $\mathfrak{Q}_3$ vanishes, while for minimal coupling (and any spacetime curvature) we have
\be
\label{eqn:clinemin}
\langle \nord{\rho_U} \circ \gamma \rangle_\omega (f^2) \geq - \left[ \int_0^\infty \frac{d\alpha}{\pi} \phi^*((\nabla_{U} \otimes \nabla_{U}) \, W_0)(\bar{f_\alpha},f_\alpha)+\frac{m^2}{n-2}  \langle \nord{\Phi^2}  \circ \gamma \rangle_\omega(f^2) \right] 
\,.
\ee   

There is another interesting situation in which a variant of the above result can be obtained. Suppose that the background spacetime is such that $\mathcal{R}_\xi$
is nonnegative. In particular, this occurs if the background solves the Einstein equations with matter that obeys both the strong and weak energy conditions, for then we have $R_{\mu \nu} U^\mu U^\nu\le 0$ and $R\ge 2R_{\mu\nu} U^\mu U^\nu$, whereupon
\begin{equation}
\frac{2\xi}{n-2} R -R_{\mu \nu} U^\mu U^\nu \ge -\left(1-\frac{4\xi}{n-2}\right) R_{\mu \nu} U^\mu U^\nu \ge -\left(1-\frac{2}{n-1}\right) R_{\mu \nu} U^\mu U^\nu \ge 0
\end{equation}
if $\xi$ obeys the standing assumption $\xi\in  \left[0,2\xi_c\right]$ and the spacetime dimension $n\ge 3$. Making the further mild technical assumption that $\mathcal{R}_\xi$ has a smooth nonnegative square root,\footnote{Not all smooth nonnegative functions have smooth square roots; see~\cite{Glaeser:1963,BonyColombiniPernazza:2010}.} the corresponding contributions to the averaged EED in \eqref{eqn:rhoquantum} are of the same form as those in $\hat{\rho}_1$ and can be treated in the same way. In this situation, the QEI becomes
\be \label{eqn:qboundline_mod}
\langle \nord{\rho_U}  \circ \gamma \rangle_\omega (f^2) \geq - \left( \mathfrak{Q}_1[f] + \mathfrak{Q}_4[f] +( \nord{\Phi^2} \circ \gamma )_\omega (\mathfrak{Q}_2[f]) -\xi\langle \nabla_A \nord{\Phi^2} \circ \gamma \rangle_\omega (f^2)\right) \,,
\ee
where
\be
\mathfrak{Q}_4[f]=\xi \int_0^\infty \frac{d\alpha}{\pi}  \left(\phi^*( (\sqrt{\mathcal{R}_\xi}\otimes\sqrt{\mathcal{R}_\xi}) \, W_0)\right)(\bar{f_\alpha},f_\alpha)  \,.
\ee
Although this bound is still state-dependent, the coefficients appearing in the state-dependent parts no longer depend explicitly on the background geometry. If the background spacetime solves Einstein equations with matter satisfying the strong energy condition (but not necessarily the weak energy condition) then a similar procedure could be used to absorb the $-R_{\mu\nu}U^\mu U^\nu$ term leaving the Ricci scalar in the state-dependent part.
 
\subsection{Worldvolume}
\label{sub:qvolume}

In this subsection we will consider averages of the EED over a spacetime volume. We will require some more terminology and notation for these purposes. 
First, following Ref.~\cite{Fewster:2007rh} we define a small sampling domain to be an
open subset\footnote{Ref.~\cite{Fewster:2007rh} allows for $\Sigma$ to be a timelike submanifold of dimension lower than $n$, but we will not need that level of generality here.} $\Sigma$ of $(M,g)$ that (i) is contained in a globally hyperbolic convex normal neighbourhood of $M$, and (ii) may be covered by a single \emph{hyperbolic coordinate chart} $\{ x^a \}$, which, by definition, requires that $\partial/ \partial x^0$ is future pointing and timelike and that there exists a constant $c>0$ such that 
\be \label{eqn:speed}
c|u_0| \geq \sqrt{\sum_{j=1}^{n-1} u_j^2} 
\ee
holds for the components of every causal covector $u$ at each point of $\Sigma$ -- in other words, the coordinate speed of light is bounded. Now we may express the hyperbolic chart $\{ x^\mu \}$ by a map $\kappa$ where $\Sigma \to \mathbb{R}^n$, $\kappa (p)= (x^0(p),\ldots,x^{n-1}(p))$. Any function $f$ on $\Sigma$ determines a function $f_\kappa = f\circ \kappa^{-1}$ on $\Sigma_\kappa=\kappa(\Sigma)$. In particular, the inclusion map $\iota:\Sigma\to M$ induces a smooth map $\iota_\kappa: \Sigma_\kappa \to M$. Then the bundle $\mathcal{N}^+$ of non-zero future pointing null covectors on $(M,g)$ pulls back under $\iota_\kappa$ so that
\be
\iota^*_\kappa \mathcal{N}^+ \subset \Sigma_\kappa \times \Gamma \,,
\ee
where $\Gamma \subset \mathbb{R}^n$ is the set of all $u_a$ with $u_0 >0$ and satisfying Eq.~(\ref{eqn:speed}) so it is a proper subset of the upper half space $\mathbb{R}^+ \times \mathbb{R}^{n-1}$. For brevity, if $S$ is a smooth function on $M\times M$, we write $S_\kappa$ instead of $S_{\kappa\times\kappa}$ for $S\circ(\kappa^{-1}\times\kappa^{-1})$.

\subsubsection{Bound with explicit mass-dependence}

Let $f$ be any real-valued test function  compactly supported in the small sampling domain $\Sigma$ and let $U^\mu$ be a future-directed timelike unit vector field defined on a neighbourhood of the support of $f$. Applying the
Gram-Schmidt process to the basis $U,\partial/\partial x^1,\ldots,\partial/\partial x^{n-1}$, we obtain a smooth $n$-bein $\{e_a^\mu\}_{a=0,1 \dots n-1}$ on this neighbourhood, with $e_0^\mu=U^\mu$. 

Following a procedure similar to the one used to derive the worldline inequality, we fix a Hadamard reference state $\omega_0$ with $2$-point function $W_0$. Then the expectation values of the effective energy density in Hadamard state $\omega$ and normal-ordered relative to $\omega_0$, can be written using Eq.~(\ref{eqn:EED1b}) as
\begin{equation}\label{eqn:rhoquantumworld}
\langle \nord{\rho_U} (f^2) \rangle_\omega =  \int \, \dV f^2 \coin{ \hat{\rho}^{\RN{1}}\nord{W} } 
-\left \langle\nord{\Phi^2}\left(\mathfrak{Q}^{\RN{1}}_2[f]\right)\right\rangle_\omega\,,
\end{equation}
where the operator $\hat{\rho}^{\RN{1}}$ (the superscript $\RN{1}$ merely serves to distinguish this bound from a bound without explicit mass dependence that will be described shortly) is given by 
\be
\hat{\rho}^{\RN{1}}=\left(1-\frac{2\xi}{n-2}\right)(U^\mu \nabla_\mu \otimes U^\nu \nabla_\nu)+\frac{2\xi}{n-2} \sum_{a=1}^{n-1} \left( e^\mu_a \nabla_\mu \otimes e^\nu_a \nabla_\nu\right) 
\ee
and
\begin{equation}\label{eq:QAU}
\mathfrak{Q}^{\RN{1}}_2[f] = \xi \nabla_\mu\nabla_\nu ( f^2 U^\mu U^\nu)+ \frac{1-2\xi}{n-2} m^2 f^2  -\xi\mathcal{R}_\xi f^2\,.
\end{equation}

All terms appearing in $\hat{\rho}^{\RN{1}}$ take the form $Q\otimes Q$, with $Q$ a partial differential operator with smooth real coefficients, provided that the coupling obeys $\xi \in [0,(n-2)/2]=[0,2(n-1)\xi_c]$. Their contributions to the averaged EED can then all be bounded from below using the following result, which is similar to Lemma \ref{lem:Qline} for the worldline. 
\begin{lemma}
	\label{lem:Qworl}
	Let $Q$ be a partial differential operator on $M$ with smooth real coefficients.
	Then the inequality
	\begin{equation}
		\label{eqn:lemworl}
		 \int \, \dV  f^2 \coin{(Q\otimes Q)\nord{W} }  \ge  - 2 \int_{\mathbb{R}^+ \times \mathbb{R}^{n-1}} \frac{d^n\alpha}{(2\pi)^n} \left((Q\otimes Q) W_0\right)_\kappa (\overline{h_\alpha}, h_\alpha) > -\infty\,,
	\end{equation}  
	holds for all Hadamard states $\omega$ and real-valued test functions $f$ supported in the small sampling domain $\Sigma$, where $\alpha=(\alpha_0,\ldots,\alpha_{n-1})$ and 
	\begin{equation} 
	h_\alpha(x)=e^{i\alpha_\mu x^\mu} (-g_\kappa(x))^{1/4} f_\kappa(x).
	\end{equation}  
\end{lemma}
\begin{proof}
	The proof works in a similar way to that of Lemma \ref{lem:Qline}. First note that, for any smooth symmetric function $S$ and test function $f$ as above, and writing 
	$h(x)=  (-g_\kappa(x))^{1/4} f_\kappa(x)$, we have 
	\begin{align}
		\int \, \dV f^2 \coin{S}&= \int d^n x\,d^n x'   \delta^n(x-x')h(x)h(x')   S_\kappa(x,x') \nonumber \\
		&= \int_{\mathbb{R}\times \mathbb{R}^{n-1}} \frac{d^n \alpha}{(2\pi)^n}
		\int d^n x \,d^n x' e^{-i\alpha(x-x')}h(x)h(x')   S_\kappa(x,x') \nonumber \\
		&= 2 \int_{\mathbb{R}^+ \times \mathbb{R}^{n-1}} \frac{d^n \alpha}{(2\pi)^n}
		\int d^n x\,d^n x' \, e^{-i\alpha(x-x')}h(x)h(x')   S_\kappa(x,x') \nonumber\\
		&= 2 \int_{\mathbb{R}^+ \times \mathbb{R}^{n-1}} \frac{d^n \alpha}{(2\pi)^n} S_\kappa(\overline{h_\alpha}, h_\alpha)\,,
	\end{align}
	where we have inserted the Fourier representation of the $\delta$-function in the second step and in the last step we used the symmetry of $S$. All these manipulations are valid owing to the compact support of $f$ (and hence $h$) and the smoothness of $f$ and $S$. Applying this to $S=(Q\otimes Q) \nord{W}$, gives
	\begin{align}
		\text{L.H.S. of~\eqref{eqn:lemworl}} &= 2\int_{\mathbb{R}^+ \times \mathbb{R}^{n-1}} \frac{d^n \alpha}{(2\pi)^n} \left((Q\otimes Q) \nord{W} \right)_\kappa(\overline{h_\alpha}, h_\alpha) \nonumber\\ 
		&=2 \int_{\mathbb{R}^+ \times \mathbb{R}^{n-1}} \frac{d^n \alpha}{(2\pi)^n} \left\{\left((Q\otimes Q)  W \right)_\kappa (\overline{h_\alpha}, h_\alpha)- \left((Q\otimes Q)  W_0\right)_\kappa (\overline{h_\alpha}, h_\alpha) \right\}\,.
	\end{align} 
	In the last line, the two terms in the integrand involve pull-backs of distributions via $\kappa\times\kappa$. The criteria for existence are trivially satisfied in this case, and their wave-front sets are bounded by
	\begin{equation}
	\WF((Q\otimes Q)  W )_\kappa) \subset (\iota^*_\kappa\times \iota^*_\kappa) (\mathcal{N}^+\times  \mathcal{N}^-) \subset (\Sigma_\kappa \times \Gamma) \times (\Sigma_\kappa \times -\Gamma)
	\end{equation}
	and the same bound for $W_0$. As $\Gamma$ is contained in the $\alpha_0>0$ half-space of $\mathbb{R}^n$, the two terms in the integrand are separately rapidly decaying as $\alpha\to\infty$ in the integration region (see~\cite[p.~444]{Fewster:2007rh}) Therefore the integrals exist separately; as $((Q\otimes Q) W)$ is non-negative, we may discard this term, whereupon the inequality~\eqref{eqn:lemworl} is proved.
\end{proof}
Applying this lemma to Eq.~\eqref{eqn:rhoquantumworld}, we have proved the following result for $\xi \in [0,(n-2)/2]$, a range that includes the conformal coupling $\xi_c$:
\begin{theorem}\label{thm:qvol1}
For coupling constant $\xi \in [0,2(n-1)\xi_c]$, suppose that $f$ is a real-valued test function compactly supported in a small sampling domain and let $U^\mu$ be a future-pointing unit timelike vector field defined on a neighbourhood of the support of $f$. Then the QEI
\begin{align} \label{eqn:qboundone}
\langle\nord{\rho_U} (f^2) \rangle_\omega &\geq - \left(\mathfrak{Q}_1^{\RN{1}}[f]
+\left \langle\nord{\Phi^2}\left(\mathfrak{Q}^{\RN{1}}_2[f]\right)\right\rangle_\omega  \right)
  \,, 
\end{align}
where 
\be
\mathfrak{Q}_1^{\RN{1}}[f]= 2\int_{\mathbb{R}^+ \times \mathbb{R}^{n-1}} \frac{d^n \alpha}{(2\pi)^n}  (\hat{\rho}^{\RN{1}} \, W_0) (\overline{h_\alpha},h_\alpha) \,,
\ee
holds for all Hadamard states $\omega$, where $W_0$ is the $2$-point function of the Hadamard reference state used to define the normal-ordering prescription, and $\hat{\rho}^{\RN{1}}$, $h_\alpha$, and $\mathfrak{Q}^{\RN{1}}_2[f]$ are as defined above.
\end{theorem}
As in Sec.~\ref{sub:qline} there are situations in which some of the state-dependent terms in this bound may be bounded independently of the state. In particular, the mass term can be treated in this way if $\xi\in [1/2,2(n-1)\xi_c]$, and the $\mathcal{R}_\xi$ term can if the background obeys
SEC and WEC and $\xi\in [0,(n-1)\xi_c]$ (and we assume that $\sqrt{\mathcal{R}_\xi}$ is smooth).
There are also obvious simplifications for Ricci flat spacetimes and at minimal coupling $\xi=0$.
We leave the details to the reader.

\subsubsection{Bound without explicit mass dependence}

The bound derived depends on the mass of the field but, as in the classical case~\cite{Brown:2018hym}, we can derive a second bound that does not have explicit mass dependence by using the field equation.

Instead of starting with~\eqref{eqn:EED1}, we use Eq.~(\ref{eqn:EED2}), with identities~\eqref{eq:Wick_id1} and \eqref{eq:Wick_id2}, which gives
\begin{equation}\label{eqn:rhoquantumworld2}
	 \langle \nord{\rho_{U}} (f^2) \rangle_\omega =  \int \, \dV \coin{ \hat{\rho}^{\RN{2}} \nord{W} }  -
	 \left \langle\nord{\Phi^2}\left(\mathfrak{Q}^{\RN{2}}_2[f]\right)\right\rangle_\omega 
	 \,,
\end{equation} 
where
\be
\hat{\rho}^{\RN{2}}=\frac{n-3}{n-2}(U^\mu \nabla_\mu \otimes U^\nu \nabla_\nu)+\frac{1}{n-2} \sum_{a=1}^{n-1} ( e^\mu_a \nabla_\mu \otimes e^\nu_a \nabla_\nu) 
\ee
and 
\begin{equation}\label{eq:QBU}
\mathfrak{Q}^{\RN{2}}_2[f] = \xi \nabla_\nu\nabla_\mu (f^2 U^\mu U^\nu )
-\frac{1-2\xi}{2(n-2)} \Box_g f^2  -\xi \mathcal{R}_{1/2} f^2 \,.
\end{equation}

Applying Lemma~\ref{lem:Qworl} to the terms arising from $\hat{\rho}^{\RN{2}}$, we have proved:
\begin{theorem}\label{thm:qvol2}
	For coupling constant $\xi \in \mathbb{R}$, suppose that $f$ is a real-valued test function compactly supported in a small sampling domain and let $U^\mu$ be a future-pointing unit timelike vector field defined on a neighbourhood of the support of $f$. Then the QEI 
\begin{equation} \label{eqn:qboundtwo}
\langle \nord{\rho_U} (f^2) \rangle_\omega  \geq -  \left(\mathfrak{Q}^{\RN{2}}_1[f]
+\left \langle\nord{\Phi^2}\left(\mathfrak{Q}^{\RN{2}}_2[f]\right)\right\rangle_\omega\right) \,,
\end{equation}  
holds for all Hadamard states $\omega$, where
\be
\mathfrak{Q}^{\RN{2}}_1[f]=2 \int_{\mathbb{R}^+ \times \mathbb{R}^{n-1}} \frac{d^n \alpha}{(2\pi)^n}   (\hat{\rho}^{\RN{2}} \, W_0) (\overline{h_\alpha},h_\alpha)  \,,
\ee 
$W_0$ is the $2$-point function of the Hadamard reference state used to define the normal-ordering prescription,
and $\hat{\rho}^{\RN{2}}$, $h_\alpha$, and $\mathfrak{Q}^{\RN{2}}_2[f]$ are as defined above. 
\end{theorem} 
As with Theorem~\ref{thm:qvol1} there are situations in which some of the state-dependent terms may be given state-independent lower bounds, and also simplifications at minimal coupling and in Ricci-flat spacetimes. Details are left to the reader. 
For $\xi\in [0,2(n-1)\xi_c]$ Theorems~\ref{thm:qvol1} and~\ref{thm:qvol2} may be combined into a single theorem by taking the stricter of the two bounds in each case. This would be similar to the classical Theorem 2 of Ref.~\cite{Brown:2018hym}.

It is interesting to note that we can write the first term on the R.H.S. of Eqs.~\eqref{eq:QAU} and \eqref{eq:QBU} using
\begin{equation}
\nabla_\nu\nabla_\mu (f^2 U^\mu U^\nu) = 
f^2\left(\nabla_{\mu} A^\mu + \theta^2 +\nabla_U\theta\right) + \nabla_U^2 f^2 + 2\theta\nabla_Uf^2+ \nabla_A f^2  
\end{equation}
where $\theta=\nabla_{\mu}U^\mu$ is the expansion and $A^\mu$ is the acceleration. If $U^\mu$ is an irrotational timelike geodesic congruence then the Raychaudhuri equation
	\be
\label{eqn:ray0}
\nabla_U\theta =R_{\mu \nu} U^\mu U^\nu-2\sigma^2 -\frac{\theta^2}{n-1}  
\ee
gives  
\begin{equation}
\nabla_\nu\nabla_\mu (f^2 U^\mu U^\nu) = 
f^2\left( R_{\mu\nu} U^\mu U^\nu +\frac{n-2}{n-1}\theta^2 - 2\sigma^2 \right) + \nabla_U^2 f^2 + 2\theta\nabla_Uf^2  \,,
\end{equation}
where $\sigma$ is the shear scalar. If the background obeys the SEC then $-R_{\mu\nu} U^\mu U^\nu$ is nonnegative and, together with the shear term $-2\sigma^2$, can be absorbed into the state-independent part of the bound (assuming they have smooth square roots).

\section{Minkowski space}

To illustrate our results let us consider the non-minimally coupled scalar field in the $n$-dimensional Minkowski space $M_{\text{mink}} = (\mathbb{R}^n, \eta)$, with $\eta=\text{diag}(+1,-1,\ldots,-1)$ being the standard Minkowski metric.  

\subsection{Worldline}
\label{sec:mink-space}

First we will apply the worldline bound of Theorem \ref{the:linegen} to $n$-dimensional Minkowski space for the case of an inertial curve, i.e., $\gamma$ is a timelike geodesic. Without loss of generality, inertial coordinates may be chosen so that $\gamma(t)=(t,\bm{0})$. Also we choose our reference state to be the vacuum state $\Omega$, which has the two-point function
\be
\label{eqn:vac}
W_\Omega(t,\bm{x},t',\bm{x}')
=\int d\mu (\bm{k}) e^{-ik_\mu (x-x')^\mu}
=\int d\mu (\bm{k}) e^{-i[(t-t')\omega(\bm{k})+(\bm{x}-\bm{x}')\cdot \bm{k}]} \,,
\ee
where 
\be
\label{eqn:measure}
d\mu(\bm{k})=\int \frac{d^{n-1}\bm{k} }{(2\pi)^{n-1}}\frac{1}{2\omega(\bm{k})} \,,
\ee
is the measure and $k_\mu=(\omega(\bm{k}),\bm{k})$ with $\omega(\bm{k})=\sqrt{\bm{k}^2+m^2}$.
The operator $\hat{\rho}_1$ from Eq.~(\ref{eqn:lineoper}) can be written as
\be
\label{eqn:rho1mink}
\hat{\rho}_1=\left(1-2\xi\frac{n-1}{n-2}\right)(\partial_0 \otimes \partial_{0}) +\frac{2\xi}{n-2} \sum_{i=1}^{n-1} (\partial_i \otimes \partial_{i}) \,.
\ee
If we use the identity
\be
m^2 \phi^* W_\Omega (\bar{g} \otimes g)+\sum_{i=1}^{n-1} \phi^* (\partial_i \otimes \partial_i) W_\Omega (\bar{g} \otimes g)=\phi^* (\partial_0 \otimes \partial_0) W_\Omega (\bar{g} \otimes g) \,,
\ee
we have for the first part of the bound
\bea \label{eqn:minklineone}
\mathfrak{Q}_1[f]&=&\int_0^\infty \frac{d\alpha}{\pi} \left( \phi^* (\hat{\rho}_1 W_\Omega)(\bar{f}_\alpha,f_\alpha)+2\xi \alpha^2 \phi^* W_\Omega (\bar{f}_\alpha,f_\alpha) \right)\nonumber\\
&=&\int_0^\infty \frac{d\alpha}{\pi}\left( \left( (1-2\xi)(\partial \otimes \partial)-\frac{2\xi}{n-2} m^2+2\xi \alpha^2\right)\phi^* W_\Omega\right) (\bar{f}_\alpha,f_\alpha)\nonumber\\
&=& \int_0^\infty \frac{d\alpha}{\pi} \int dt \, dt' e^{-i \alpha (t-t')} f(t) f(t') \times \nonumber\\
&& \qquad \qquad \int d\mu (\bm{k}) \left( \omega^2(\bm{k}) (1-2\xi)-\frac{2\xi}{n-2} m^2+2\xi \alpha^2\right) e^{-i (t-t')\omega(\bm{k})} \nonumber\\
&=& \frac{S_{n-2}}{(2\pi)^n} \int_0^\infty d\alpha \int_0^\infty dk \frac{k^{n-2}}{\omega(k)} \left( \omega^2(k) (1-2\xi)-\frac{2\xi}{n-2} m^2+2\xi \alpha^2\right) |\hat{f}(\alpha+\omega(k))|^2 \,. \nonumber
\eea
At the last step we have passed to spherical polar coordinates and written $S_{n-2}$ for the volume of the $(n-2)$-dimensional standard unit sphere. Our convention for the Fourier transform of $f$ is
\be
\hat{f}(\omega)=\int dt \, e^{i\omega t} f(t) \,.
\ee

We can make the change of variables
\be
u=\alpha+\omega(k) \,, \qquad v=\omega(k) \,,
\ee
and write $\mathfrak{Q}_1 [f]$ as 
\bea \label{eqn:oaq}
\mathfrak{Q}_1 [f]= \frac{S_{n-2}}{(2\pi)^n} \int_m^\infty du |\hat{f}|^2(u) \int_m^u dv (v^2-m^2)^{(n-3)/2} \left(v^2-4\xi u v+2 \xi\left(u^2-\frac{1}{n-2}m^2\right)\right) \,. \nonumber \\
\eea
Using Eq.~(\ref{eqn:minklineone}) and the fact that $\mathfrak{Q}_3$ vanishes at flat spacetime Theorem \ref{the:linegen} becomes 

\begin{theorem}
\label{the:minkline}
In n-dimensional Minkowski space we have, for $0\le\xi\le 2\xi_c$,  
\be
\label{eqn:linemink}
\langle \nord{\rho_U} \circ \gamma\rangle_\omega (f^2) \geq -\left[\mathfrak{Q}_1[f] \mathbb{1}+\langle \nord{\Phi^2} \circ \gamma \rangle_\omega (\mathfrak{Q}_2[f]) \right] \,,
\ee
where $\gamma$ is a timelike geodesic,
\bea
\mathfrak{Q}_1 [f]&=& \frac{S_{n-2}}{(2\pi)^n} \int_m^\infty du\, u^n |\hat{f}|^2 (u) \times \\
&&\qquad \left( \frac{1}{n} Q_{n,2} \left(\frac{u}{m}\right)- \frac{4\xi}{n-1} Q_{n,1} \left(\frac{u}{m}\right)+ \frac{2\xi}{n-2} \left(1-\frac{m^2}{u^2(n-2)}\right) Q_{n,0} \left(\frac{u}{m}\right) \right)\nonumber \,,
\eea
and
\be
\mathfrak{Q}_2[f](t)=\frac{1-2\xi}{n-2} m^2 f^2(t) +2\xi  ( f'(t))^2 \,,
\ee
for $f\in \mathcal{D}(\mathbb{R})$ and $\xi \in \left[0,\frac{n-2}{2(n-1)}\right]$.
The functions $Q_{n,r}$ are defined by
\be \label{eqn:qdef}
Q_{n,r}(y)= \frac{n+r-2}{y^{n+r-2}} \int_1^y dx (x^2-1)^{(n-3)/2} x^r \,.
\ee
\end{theorem}

The functions $Q_{n,r}$ are non-negative functions for $n+r \geq 2$ and they vanish for $n+r=2$. Also they have the properties 
\be
Q_{n,r} (1)=0\,, \qquad \lim_{y\to \infty} Q_{n,r}(y)=1 \,.
\ee
Since $n\geq 3$ and $\xi \leq \frac{n-2}{2(n-1)}$ 
\be \label{eqn:linelim}
\mathfrak{Q}_1 [f] \leq \frac{2S_{n-2}}{(2\pi)^n(n-1)} \int_0^\infty du \, u^n |\hat{f}|^2 (u)  \,,
\ee
where we also assumed that $m>0$.

Using Eq.~(\ref{eqn:linelim}) we can investigate the behavior of the bound under rescaling of the smearing function $f$, and in particular whether the SEC holds in an averaged sense along a complete timelike geodesic. This question is prompted by the analogous situation for the energy density, in which an averaged weak energy condition (AWEC) can be proved for the nonminimally coupled scalar field under mild assumptions on the growth of the Wick square along the geoedesic~\cite{Fewster:2007ec}. In fact, we will not be able to prove a direct analogue of the AWEC result, but instead a slight modification of it. 

First we define the smearing function $f_\lambda$ for $\lambda \in \mathbb{R}$ to be
\be
f_\lambda(t)=\frac{f(t/\lambda)}{\sqrt{\lambda}} \,.
\ee
Then its Fourier transform satisfies
\be 
\left( \hat{f}_\lambda (u) \right)^2=\lambda \left(\hat{f}(\lambda u) \right)^2 \,,
\ee
and, by analogy with the averaged weak energy condition (AWEC)~\cite{Fewster:2007ec}, we would say that the averaged strong energy condition (ASEC) holds in state $\omega$ if 
\begin{equation}
\liminf_{\lambda\to +\infty}\lambda  \langle \nord{\rho_U} \circ \gamma \rangle_\omega (f_\lambda^2) \ge 0
\end{equation}
because the left-hand side is a measure of the total integral of the EED (up to a factor of $f(0)$) along $\gamma$. 

Using Eq.~(\ref{eqn:linelim}) we have that $\mathfrak{Q}_1 [f_\lambda]=O(\lambda^{-n})$ as $\lambda \to \infty$. Then $\lambda \mathfrak{Q}_1 (f_\lambda) \to 0$ as $\lambda \to \infty$. For the state dependent part of the bound we have
\be
\lambda\langle \nord{\Phi^2} \circ \gamma \rangle_\omega (\mathfrak{Q}_2[f_\lambda])
= \frac{1-2\xi}{n-2} m^2 \int dt\, \langle \nord{\Phi^2} \rangle_\omega (\gamma(t)) f (t/\lambda)^2 +\frac{2\xi}{\lambda^2} \int dt\, \langle \nord{\Phi^2} \rangle_\omega (\gamma(t)) f' (t/\lambda)^2 \,.
\ee
If we assume that $ \langle \nord{\Phi^2} \rangle_\omega (t)$ is absolutely integrable, then the dominated convergence theorem implies that the first term converges to a constant while the second term goes to zero for $\lambda \to \infty$. This gives a bound 
\begin{equation}
\liminf_{\lambda\to\infty}\lambda  \langle \nord{\rho_U} \circ \gamma \rangle_\omega (f_\lambda^2) \ge -\frac{1-2\xi}{n-2} m^2 f(0)^2 \int dt\, \langle \nord{\Phi^2} \rangle_\omega (\gamma(t))
\end{equation}
rather than the ASEC in the form originally stated. This does not show that ASEC
cannot hold, but rather that it cannot be derived from the QSEI by scaling methods.
Instead, what can be proved is that 
\begin{equation}
\liminf_{\lambda\to +\infty}\lambda  \left\langle \left(\nord{\rho_U}+ 
\frac{1-2\xi}{n-2} m^2 \nord{\Phi^2}  \right)
 \circ \gamma \right\rangle_\omega (f_\lambda^2) \ge 0,
\end{equation}
provided that $|\langle \nord{\Phi^2} \rangle_\omega (\gamma(t))|\le c (1+|t|)^{1-\epsilon}$ for positive constants $c$, $\epsilon$ (we may, without loss, assume that $0<\epsilon<1$). The proof is simple: for $\lambda\ge 1$ we have
\begin{align}
\lambda \left\langle \left(\nord{\rho_U}+ 
\frac{1-2\xi}{n-2} m^2 \nord{\Phi^2}  \right)
\circ \gamma \right\rangle_\omega (f_\lambda^2) &\ge -\frac{2\xi}{\lambda^2} \int dt\, \langle \nord{\Phi^2} \rangle_\omega (\gamma(t)) f' (t/\lambda)^2\notag\\
&\ge -\frac{2c\xi}{\lambda^\epsilon} \int du\, (1+|u|)^{1-\epsilon} f' (u)^2  \,,
\end{align} 
where we have changed variables from $t$ to $u=t/\lambda$  and used the fact that 
 $|\langle \nord{\Phi^2} \rangle_\omega (\gamma(\lambda u))|\le c\lambda^{1-\epsilon} (1+|u|)^{1-\epsilon}$ if $\lambda\ge 1$. As the right-hand side vanishes in the limit $\lambda\to+\infty$ the result follows. It would be interesting to determine whether
 the ASEC is actually violated in some Hadamard states, but we do not pursue this here.

\subsection{Worldvolume}

Now we turn to the worldvolume quantum inequalities of Theorems \ref{thm:qvol1} and \ref{thm:qvol2}. Again we choose our reference state to be the vacuum state with two-point function given by Eq.~(\ref{eqn:vac}). Additionally,  we require the vector field $U^\mu$ to be translationally invariant. Then we can choose an inertial coordinate system for which $U^\mu$ is purely in the direction of $t$. We suppress the distinction between $f$ and its coordinate expression so $f$ becomes identical with $h$. 
In Minkowski space the operators $\hat{\rho}^\RN{1}$ and $\hat{\rho}^\RN{2}$ become 
\be
\hat{\rho}^\RN{1}=\left(1-\frac{2\xi}{n-2}\right) (\partial_0 \otimes \partial_0)+{\sum_{i=1}^{n-1}}\frac{2\xi}{n-2}(\partial_i \otimes \partial_i) \,.
\ee
and
\be
\hat{\rho}^\RN{2}=\left(\frac{n-3}{n-2}\right) (\partial_0 \otimes \partial_0)+{\sum_{i=1}^{n-1}}\frac{1}{n-2}(\partial_i \otimes \partial_i) \,,
\ee
Then the state independent part of Theorem \ref{thm:qvol1} becomes
\bea
\mathfrak{Q}_1^\RN{1}[f]&&= 2 \int_{\mathbb{R}^+ \times \mathbb{R}^{n-1}} \frac{d^n \alpha}{(2\pi)^n} (\hat{\rho}^\RN{1} \, W_\Omega)(\bar{f}_\alpha,f_\alpha)=\nonumber\\ 
&&= \int_{\mathbb{R}^+ \times \mathbb{R}^{n-1}} \frac{d^n \alpha}{(2\pi)^n} \int \frac{d^{n-1} \bm{k}}{(2\pi)^{n-1}} \frac{1}{\omega(\bm{k})} \left[ \omega(\bm{k})^2-\frac{2\xi}{n-2} m^2 \right]  |\hat{f}(\alpha+k)|^2 \,, 
\eea
where we used $U^\mu k_\mu= \omega(\bm{k})$ and the convention
\begin{equation}
\hat{f}(k)=\int_{\mathbb{R}^{n}} d^nx\, e^{ik_\mu x^\mu}f(x).
\end{equation}
With the change of variables $(\bm{k},\bm{\alpha})\to(\bm{k},\bm{u})$, where
\be \label{eqn:varone}
\bm{u}=\bm{\alpha}+\bm{k} \,,
\ee
$\mathfrak{Q}^\RN{1}_1[f]$ becomes
\bea
\mathfrak{Q}_1^\RN{1}[f]&=& \int_0^\infty \frac{d \alpha_0}{(2\pi)^n} \int \frac{d^{n-1} \bm{k}}{(2\pi)^{n-1}} \frac{1}{\omega(\bm{k})} \left[ \omega(\bm{k})^2-\frac{2\xi}{n-2} m^2 \right] \nonumber\\
&& \qquad \qquad \qquad \times 
\int_{\mathbb{R}^{n-1}} d^{n-1} \bm{u} |\hat{f}(\alpha_0+\omega(\bm{k}),\bm{u})|^2 \,.
\eea
Performing a second change of variables 
\be
u_0=\alpha_0+\omega(\bm{k}) \,,
\ee
and changing the order of integration gives
\bea
\mathfrak{Q}_1^\RN{1}[f]&=&  \int \frac{d^{n-1} \bm{k}}{(2\pi)^{2n-1}} \frac{1}{\omega(\bm{k})} \left[ \omega(\bm{k})^2-\frac{2\xi}{n-2} m^2 \right]  \int_{\omega(\bm{k})}^\infty du_0 \int_{\mathbb{R}^{n-1}} d^{n-1} \bm{u}  |\hat{f}(u)|^2 \nonumber\\
 &=& \frac{S_{n-2}}{(2\pi)^{2n-1}} \int d k\frac{k^{n-2}}{\omega(k)} \left[ \omega(k)^2-\frac{2\xi}{n-2} m^2 \right] \int_{\omega(k)}^\infty du_0 \int_{\mathbb{R}^{n-1}} d^{n-1} \bm{u}  |\hat{f}(u)|^2 \,,
\eea
where we transitioned to spherical coordinates, using the fact that the integrand is spherically symmetric in $\bm{k}$. Writing $k$ in terms of $\omega$ gives
\bea
\mathfrak{Q}_1^\RN{1}[f]&=& \frac{S_{n-2}}{(2\pi)^{2n-1}} \int_m^\infty du_0  \int_{ \mathbb{R}^{n-1}} d^{n-1} \bm{u}\, |\hat{h}(u)|^2 \int_m^{u_0} d\omega\, (\omega^2-m^2)^{(n-3)/2} \nonumber\\
&&\qquad \qquad \qquad \qquad  \qquad \qquad \times \left[  \omega^2-\frac{2\xi}{n-2} m^2\right] \,.
\eea

Using the functions $Q_{n,k}$ defined in Eq.~\eqref{eqn:qdef} and noticing that $\mathcal{R}_\xi$ and $\mathcal{R}_{1/2}$ vanish in flat spacetime, the following theorem is immediate
\begin{theorem}
\label{the:minkvol1}
In $n$-dimensional Minkowski space for $\xi \in [ 0, (n-2)/2]$, the QSEI of Theorem~\ref{thm:qvol2} reduces to
\be
\langle\nord{\rho_U} (f^2) \rangle_\omega \geq -\left[\mathfrak{Q}_1^\RN{1}[f] \mathbb{1}+  \left \langle\nord{\Phi^2}\left(\mathfrak{Q}^{\RN{1}}_2[f]\right)\right\rangle_\omega \right]
\,, 
\ee
where
\bea
\mathfrak{Q}_1^\RN{1}[f]&=&\frac{S_{n-2}}{(2\pi)^{2n-1}} \int_m^\infty du_0  \int d^{n-1} \bm{u}\, |\hat{f}(u)|^2 u_0^n \nonumber\\
&&\qquad \qquad \qquad  \times \left[ \frac{1}{n}Q_{n,2} \left(\frac{u_0}{m}\right) - \frac{m^2}{u_0^2} \frac{2\xi}{(n-2)^2} Q_{n,0} \left(\frac{u_0}{m}\right) \right]  \,,
\eea
and
\be
\mathfrak{Q}^{\RN{1}}_2[f] = \xi \partial_0^2(f^2) + \frac{1-2\xi}{n-2} m^2 f^2 \,.
\ee
\end{theorem}

Similarly the state independent part of Theorem \ref{thm:qvol2} becomes
\bea
\mathfrak{Q}_1^\RN{2}[f]&=& 2 \int_{\mathbb{R}^+ \times \mathbb{R}^{n-1}} \frac{d^n \alpha}{(2\pi)^n} (\hat{\rho}^\RN{2} \, W_\Omega)(\bar{f}_\alpha,f_\alpha) \nonumber\\
&&= \int_{\mathbb{R}^+ \times \mathbb{R}^{n-1}} \frac{d^n \alpha}{(2\pi)^n} \int \frac{d^{n-1} \bm{k}}{(2\pi)^{n-1}} \frac{1}{\omega(\bm{k})} \left[ \omega(\bm{k})^2-\frac{1}{n-2} m^2 \right]  |\hat{f}(\alpha+k)|^2 \,, 
\eea
and one may show
\begin{theorem}
	\label{the:minkvol2}
	In $n$-dimensional Minkowski space for $\xi \in \mathbb{R}$, the QSEI of Theorem~\ref{thm:qvol2} reduces to
	\bea
	\langle\nord{\rho_U} (f^2) \rangle_\omega &\geq& - \left[\mathfrak{Q}_1^\RN{2}[f]\mathbb{1}+ \left \langle\nord{\Phi^2}\left(\mathfrak{Q}^{\RN{2}}_2[f]\right)\right\rangle_\omega \right]
	\,, 
	\eea
	where 
	\bea
	\mathfrak{Q}_1^\RN{2}[f]&=& \frac{S_{n-2}}{(2\pi)^{2n-1}} \int_m^\infty du_0  \int d^{n-1} \bm{u} |\hat{f}(u)|^2 u_0^n \nonumber\\
	&&\qquad \qquad \qquad \qquad  \times \left[ \frac{1}{n} Q_{n,2} \left(\frac{u_0}{m}\right) -\frac{m^2}{u_0^2} \frac{1}{(n-2)^2} Q_{n,0} \left(\frac{u_0}{m}\right) \right]  \,,
	\eea
	and
	\be
\mathfrak{Q}^{\RN{2}}_2[f] = \xi \partial_0^2(f^2)
-\frac{1-2\xi}{2(n-2)} \Box_g f^2  \,.
	\ee
	\end{theorem}
Noting the properties of the functions $Q_{n,k}$ and extending the integration domain, both $\mathfrak{Q}_1^\RN{1}[f]$ and $\mathfrak{Q}_1^\RN{2}[f]$ can be bounded above by the expression
\begin{equation}
\frac{S_{n-2}}{(2\pi)^{2n-1}} \int_0^\infty du_0 \,  \int d^{n-1} \bm{u}\, |\hat{f}(u)|^2 \frac{u_0^n}{n} 
\end{equation}
on the ranges of $\xi$ for which Theorems~\ref{the:minkvol1} and~\ref{the:minkvol2} are valid.

\section{KMS states and temperature scaling}
\label{sec:KMS}

The QEI bounds we have obtained depend on the state, in contrast to the original QEIs, \cite{Ford:1978qya, Ford:1990id, Ford:1996er, Flanagan:1997gn, Pfenning:1997rg, Fewster:1998pu, Fewster:1999gj, Fewster:2007rh} that provide state-independent lower bounds. 
Clearly there are some very uninteresting state-dependent bounds, such as the trivial bound in which the averaged stress-energy tensor is simply bounded below by itself! It is therefore important to explain in what way our state-dependent bounds are nontrivial. The strategy we adopt follows~\cite{Fewster:2006iy,Fewster:2007ec} in which a state-dependent lower bound of the schematic form
\begin{equation}
\langle \rho(f)\rangle_\omega \ge -\langle Q(f)\rangle_\omega
\end{equation}
is regarded as nontrivial provided there are no constants $c$ and $c'$ (perhaps depending on $f$) for which
\begin{equation}
|\langle \rho(f)\rangle_\omega| \le c+ c'  \langle Q(f)\rangle_\omega 
\end{equation}
holds for all physically reasonable states $\omega$. This indicates that the lower bound is
relatively small, in comparison with the possible magnitude of quantity that is being bounded. A good way to establish nontriviality is to consider a family of states in which the averaged energy density tends to infinity more rapidly than the bound does. 

In this section we will do this by examining the behaviour of the bounds we derived for thermal states in $n$-dimensional Minkowski space, letting the temperature become large. We start with the worldline inequality of Theorem \ref{the:minkline}. 

\subsection{Worldline}

We fix inertial coordinates $(t,\bm{x})$ on $n$-dimensional Minkowski spacetime, for $n>3$, and consider 
the averaged energy density along the inertial trajectory $\gamma(\tau)=(\tau,\bm{0})$.
Let $\omega_\beta$ be the KMS state at inverse temperature $\beta$, with respect to the time parameter $t$. The state $\omega_\beta$ is Hadamard, with the two-point function
\be
W_\beta(t,\bm{x},t',\bm{x}')=\int d\mu(\bm{k}) \left( \frac{e^{-ik_\mu (x-x')^\mu}}{1-e^{-\beta\omega(\bm{k})}}+\frac{e^{ik_\mu (x-x')^\mu}}{e^{\beta\omega(\bm{k})}-1}  \right) \,,
\ee
where $\mu(\bm{k})$ is given by Eq.~(\ref{eqn:measure}). 

After normal ordering with respect to the ground state of Eq.~(\ref{eqn:vac}) we find\footnote{The corresponding expression in Ref.~\cite{Fewster:2007ec} is missing the factor $\beta^{2-n}$ in one place, but the final results are correct.} (abbreviating $\langle\cdot\rangle_{\omega_\beta}$ as $\langle\cdot\rangle_\beta$)
\be \label{eqn:betaone}
\langle \nord{\Phi^2}\rangle_\beta= \beta^{2-n} B_{n,0}(\beta m) \,,
\ee
where $B_{n,r}$ is defined on $[0,\infty)$ for $r \geq 0$ by
\be 
\label{eqn:beta}
B_{n,r}(\alpha)=\frac{S_{n-2}}{(2\pi)^{n-1}} \int_\alpha^\infty dz (z^2-\alpha^2)^{(n-3)/2} \frac{z^r}{e^z-1} \,.
\ee
As $\omega_\beta$ is time-translationally invariant, the state-dependent part of the bound in Theorem \ref{the:minkline} for that state is
\be
\langle \nord{\Phi^2} \circ \gamma \rangle_{\beta} (\mathfrak{Q}_2[f])=  \beta^{2-n} B_{n,0}(\beta m) \left( \frac{1-2\xi}{n-2} m^2  || f ||^2+2 \xi ||f'||^2\right) \,.
\ee
On the other hand, the expectation value of the renormalized EED is, after a calculation,
\bea
\label{eqn:reneed}
\langle \nord{\rho_U} \rangle_{\beta} (x)&=& \int d\mu (\bm{k}) \left(\omega(\bm{k})^2-\frac{1}{n-2}m^2\right) \frac{2}{e^{\beta \omega(\bm{k})}-1}  \nonumber\\
&=& \frac{S_{n-2}}{(2\pi)^{n-1}} \int_m^\infty d\omega \, \frac{(\omega^2-m^2)^{(n-3)/2}}{e^{\beta \omega}-1} \left( \omega^2-\frac{1}{n-2}m^2\right)  \,, 
\eea
Since $\langle \nord{\rho_U} \rangle_{\beta}$ is translationally invariant, the left hand side of Eq.~\eqref{eqn:linemink} is 
\be \label{eqn:betatwo}
\langle \nord{\rho_U} \circ \gamma \rangle_{\beta} (f^2)= \left( \beta^{-n} B_{n,2}(\beta m)-\frac{m^2}{n-2} \beta^{2-n} B_{n,0} (\beta m)\right) ||f||^2 \,.
\ee
Now we can state the following theorem

\begin{theorem}
The bound given in Theorem \ref{the:minkline} is nontrivial in the sense that there do not exist constants $c$ and $c'$ such that
\be
|\langle \nord{\rho_U} \circ \gamma \rangle_{\omega} (f^2) | \leq c+c' | \mathfrak{Q}_1 (f) \mathbb{1}+\langle \nord{\Phi^2} \circ \gamma \rangle_\omega \mathfrak{Q}_2(f) | \,,
\ee
for all Hadamard states $\omega$ unless $f$ is identically zero. 
\end{theorem}

\begin{proof}
Assuming $f\not\equiv 0$, in the limit of high temperatures $\beta \to 0$ we have from Eqs.~(\ref{eqn:betaone}, \ref{eqn:betatwo}) 
\blea \label{eqn:proof}
\lim_{\beta \to 0} \beta^n \langle \nord{\rho_U} \circ \gamma \rangle_{\beta} (f^2)&=&  B_{n,2}(0) ||f||^2 >0 \\
\lim_{\beta \to 0} \beta^n \left( \mathfrak{Q}_1 (f) +\langle \nord{\Phi^2} \circ \gamma \rangle_{\beta} \mathfrak{Q}_2(f) \right) &=& 0 \,.
\elea
If the bound of Theorem \ref{the:minkline} were trivial there would exist constants $c,c'$ such that
\be
\lim_{\beta \to 0} \beta^n \langle \nord{\rho_U} \circ \gamma \rangle_{\beta} (f^2) \le   \lim_{\beta \to 0} \beta^n \left( c+c' | \mathfrak{Q}_1 (f) | +c' | \langle \nord{\Phi^2} \circ \gamma \rangle_{\beta} \mathfrak{Q}_2(f) | \right) \,.
\ee
But Eq.~(\ref{eqn:proof}) implies
\be
0<  B_{n,2}(0) ||f||^2 \le 0 \,,
\ee
which is a contradiction.

\end{proof}

\subsection{Worldvolume}

The two bounds of the worldvolume quantum inequalities of Theorems \ref{the:minkvol1} and \ref{the:minkvol2} are also state dependent.  
Evaluating them for a KMS state $\omega_\beta$ and using Eq.~(\ref{eqn:betaone}) gives
\be \label{eqn:betaonevol}
\langle \nord{\Phi^2} (\mathfrak{Q}_2^{\RN{1},\RN{2}}[f]) \rangle_\beta= \beta^{2-n} B_{n,0}(\beta m) \int  \dV \, \mathfrak{Q}_2^{\RN{1},\RN{2}}[f] \,.
\ee

The expectation value of the renormalized EED is given by Eq.~\eqref{eqn:reneed}.
So the left hand side of the inequalities of Theorems \ref{the:minkvol1} and \ref{the:minkvol2} for state $\omega_\beta$, becomes  
\be \label{eqn:betathreevol}
\langle \nord{\rho_U} (f^2) \rangle_\beta = \left( \beta^{-n} B_{n,2}(\beta m)-\frac{1}{n-2} \beta^{2-n} B_{n,0} (\beta m)\right) \int \dV f^2(x) \,.
\ee

\begin{theorem}
The bound given in Theorem~\ref{the:minkvol1}, resp., Theorem~\ref{the:minkvol2}, is nontrivial in the sense that there do not exist constants $c$ and $c'$ such that
\be 
\left|\langle \nord{\rho_U} (f^2) \rangle_\omega\right| \leq c+c'\left|\mathfrak{Q}_1^\RN{1}[f] \mathbb{1}+  \left \langle\nord{\Phi^2}\left(\mathfrak{Q}^{\RN{1}}_2[f]\right)\right\rangle_\omega \right|\,,
\ee
resp.,
\be 
\left|\langle  \nord{\rho_U} (f^2) \rangle_\omega\right| \leq c+c'\left|\mathfrak{Q}_1^\RN{2}[f] \mathbb{1}+  \left \langle\nord{\Phi^2}\left(\mathfrak{Q}^{\RN{2}}_2[f]\right)\right\rangle_\omega \right|\,,
\ee
for all Hadamard states $\omega$ unless $f$ is identically zero. 
\end{theorem}

\begin{proof}
Assuming $f\not\equiv 0$, in the limit of high temperatures $\beta \to 0$ we have from Eqs.~(\ref{eqn:betaonevol}, \ref{eqn:betathreevol}) 
\bml
\label{eqn:proofvol}
\bea
\lim_{\beta \to 0} \beta^n\langle  \nord{\rho_U} (f^2) \rangle_\beta =   B_{n,2}(0)\int d Vol f^2(x) &>&0 \\
\lim_{\beta \to 0} \beta^n \left( \mathfrak{Q}_1^{\RN{1},\RN{2}}[f] \mathbb{1}+  \left \langle\nord{\Phi^2}\left(\mathfrak{Q}^{\RN{1},\RN{2}}_2[f]\right)\right\rangle_\beta  \right) &=& 0 \,.
\elea
If the bounds of Theorems  \ref{the:minkvol1} and \ref{the:minkvol2} were trivial there would exist constants $c$ and $c'$ such that
\be
\lim_{\beta \to 0} \beta^n \langle  \nord{\rho_U}(f^2) \rangle_\beta \le  \lim_{\beta \to 0} \beta^n \left( c+c' \left| \mathfrak{Q}_1^{\RN{1},\RN{2}}[f] \mathbb{1}+  \left \langle\nord{\Phi^2}\left(\mathfrak{Q}^{\RN{1},\RN{2}}_2[f]\right)\right\rangle_\beta  \right| \right) \,.
\ee
But Eq.~(\ref{eqn:proofvol}) implies
\be
0<  B_{n,2}(0)\int d Vol f^2(x)  \le 0 \,,
\ee
which is a contradiction. 
\end{proof}

\section{Conclusions}
\label{sec:conclusions}

The main result of this paper is the derivation of state-dependent, but nontrivial, lower bounds for the renormalised effective energy density of the non-minimally coupled field, averaged either along timelike curves or over spacetime volumes. First, we discussed the quantisation of EED in the context of algebraic quantum field theory and developed a systematic framework to derive suitable differential operators, following~\cite{Hollands:2004yh}. Additionally we showed that while the quadratic Wick ordered expressions obey the Leibniz rule but not the field equation, the differences in their expectation values obey both, so the field equation can be used to simplify expressions in difference QEIs. Then we proceeded to establish both worldline and worldvolume bounds for the renormalised EED, for intervals of coupling constants including both minimal and conformal coupling  in all cases.  

Applying the results to Minkowski space we derived simplified worldline and worldvolume bounds, which are expected to hold to good approximation in circumstances where the spacetime is approximately flat or the sampling function has support that is small in comparison with curvature length scales. Finally we analysed the state dependence of the bounds in the case of $n$-dimensional Minkowski space, by looking at their temperature dependence in KMS states. We concluded that both the worldline and the worldvolume bounds are non-trivial in the sense described in the introduction. 

This is the first derivation of a quantum strong energy inequality and one of the few QEI results to address the scalar field with nonminimal coupling. More importantly, the establishment of a QSEI is the first step towards a Hawking-type singularity theorem result employing QEI hypotheses. As shown by Refs.~\cite{Fewster:2010gm} and \cite{Brown:2018hym} it is possible to prove singularity theorems of `Hawking type' if we can establish bounds of the form
\be
\label{eqn:sing}
\int R_{\mu \nu}U^\mu U^\nu f(\tau)^2 \leq ||| f|||^2 \,,
\ee
where $||| \cdot |||$ is a suitable Sobolev norm. In the case that the metric $g_{\mu \nu}$ and Hadamard state $\omega$ are physical solutions of the semiclassical Einstein equation
\be
\langle \nord{T_{\mu \nu}} \rangle_\omega =-8\pi G_{\mu \nu} \,.
\ee
the QEI bounds derived could, in some cases, be written in the geometric form of Eq.~\eqref{eqn:sing}. The investigation of this possibility and proof of a Hawking-type singularity theorem with a QEI derived hypothesis is part of an ongoing work to appear elsewhere.

\section*{Acknowledgements}

We thank Eli Hawkins and Daniel Siemssen for comments on the text and Markus Fr\"ob for useful discussions.  This work is part of a project that has received funding from the European Union's Horizon 2020 research and innovation programme under the Marie Sk\l odowska-Curie grant agreement No. 744037 ``QuEST''.

\bibliography{quant}

\begin{thebibliography}{45}%
\makeatletter
\providecommand \@ifxundefined [1]{%
 \@ifx{#1\undefined}
}%
\providecommand \@ifnum [1]{%
 \ifnum #1\expandafter \@firstoftwo
 \else \expandafter \@secondoftwo
 \fi
}%
\providecommand \@ifx [1]{%
 \ifx #1\expandafter \@firstoftwo
 \else \expandafter \@secondoftwo
 \fi
}%
\providecommand \natexlab [1]{#1}%
\providecommand \enquote  [1]{``#1''}%
\providecommand \bibnamefont  [1]{#1}%
\providecommand \bibfnamefont [1]{#1}%
\providecommand \citenamefont [1]{#1}%
\providecommand \href@noop [0]{\@secondoftwo}%
\providecommand \href [0]{\begingroup \@sanitize@url \@href}%
\providecommand \@href[1]{\@@startlink{#1}\@@href}%
\providecommand \@@href[1]{\endgroup#1\@@endlink}%
\providecommand \@sanitize@url [0]{\catcode `\\12\catcode `\$12\catcode
  `\&12\catcode `\#12\catcode `\^12\catcode `\_12\catcode `\%12\relax}%
\providecommand \@@startlink[1]{}%
\providecommand \@@endlink[0]{}%
\providecommand \url  [0]{\begingroup\@sanitize@url \@url }%
\providecommand \@url [1]{\endgroup\@href {#1}{\urlprefix }}%
\providecommand \urlprefix  [0]{URL }%
\providecommand \Eprint [0]{\href }%
\providecommand \doibase [0]{http://dx.doi.org/}%
\providecommand \selectlanguage [0]{\@gobble}%
\providecommand \bibinfo  [0]{\@secondoftwo}%
\providecommand \bibfield  [0]{\@secondoftwo}%
\providecommand \translation [1]{[#1]}%
\providecommand \BibitemOpen [0]{}%
\providecommand \bibitemStop [0]{}%
\providecommand \bibitemNoStop [0]{.\EOS\space}%
\providecommand \EOS [0]{\spacefactor3000\relax}%
\providecommand \BibitemShut  [1]{\csname bibitem#1\endcsname}%
\let\auto@bib@innerbib\@empty
\bibitem [{\citenamefont {Ford}(1978)}]{Ford:1978qya}%
  \BibitemOpen
  \bibfield  {author} {\bibinfo {author} {\bibfnamefont {L.~H.}\ \bibnamefont
  {Ford}},\ }\href {\doibase 10.1098/rspa.1978.0197} {\bibfield  {journal}
  {\bibinfo  {journal} {Proc. Roy. Soc. Lond.}\ }\textbf {\bibinfo {volume}
  {A364}},\ \bibinfo {pages} {227} (\bibinfo {year} {1978})}\BibitemShut
  {NoStop}%
\bibitem [{\citenamefont {Fewster}(2012)}]{Fewster:2012yh}%
  \BibitemOpen
  \bibfield  {author} {\bibinfo {author} {\bibfnamefont {C.~J.}\ \bibnamefont
  {Fewster}},\ }\href@noop {} {\bibfield  {journal} {\bibinfo  {journal} {ArXiv
  preprints}\ } (\bibinfo {year} {2012})},\ \Eprint
  {http://arxiv.org/abs/1208.5399} {arXiv:1208.5399 [gr-qc]} \BibitemShut
  {NoStop}%
\bibitem [{\citenamefont {Fewster}(2017)}]{Fewster2017QEIs}%
  \BibitemOpen
  \bibfield  {author} {\bibinfo {author} {\bibfnamefont {C.~J.}\ \bibnamefont
  {Fewster}},\ }\enquote {\bibinfo {title} {Quantum energy inequalities},}\ in\
  \href {\doibase 10.1007/978-3-319-55182-1_10} {\emph {\bibinfo {booktitle}
  {Wormholes, Warp Drives and Energy Conditions}}},\ \bibinfo {editor} {edited
  by\ \bibinfo {editor} {\bibfnamefont {F.~S.~N.}\ \bibnamefont {Lobo}}}\
  (\bibinfo  {publisher} {Springer International Publishing},\ \bibinfo
  {address} {Cham},\ \bibinfo {year} {2017})\ pp.\ \bibinfo {pages}
  {215--254}\BibitemShut {NoStop}%
\bibitem [{\citenamefont {Ford}\ and\ \citenamefont
  {Roman}(1996)}]{Ford:1995wg}%
  \BibitemOpen
  \bibfield  {author} {\bibinfo {author} {\bibfnamefont {L.~H.}\ \bibnamefont
  {Ford}}\ and\ \bibinfo {author} {\bibfnamefont {T.~A.}\ \bibnamefont
  {Roman}},\ }\href {\doibase 10.1103/PhysRevD.53.5496} {\bibfield  {journal}
  {\bibinfo  {journal} {Phys. Rev.}\ }\textbf {\bibinfo {volume} {D53}},\
  \bibinfo {pages} {5496} (\bibinfo {year} {1996})},\ \Eprint
  {http://arxiv.org/abs/gr-qc/9510071} {arXiv:gr-qc/9510071 [gr-qc]}
  \BibitemShut {NoStop}%
\bibitem [{\citenamefont {Pfenning}\ and\ \citenamefont
  {Ford}(1997)}]{Pfenning:1997wh}%
  \BibitemOpen
  \bibfield  {author} {\bibinfo {author} {\bibfnamefont {M.~J.}\ \bibnamefont
  {Pfenning}}\ and\ \bibinfo {author} {\bibfnamefont {L.~H.}\ \bibnamefont
  {Ford}},\ }\href {\doibase 10.1088/0264-9381/14/7/011} {\bibfield  {journal}
  {\bibinfo  {journal} {Class. Quant. Grav.}\ }\textbf {\bibinfo {volume}
  {14}},\ \bibinfo {pages} {1743} (\bibinfo {year} {1997})},\ \Eprint
  {http://arxiv.org/abs/gr-qc/9702026} {arXiv:gr-qc/9702026 [gr-qc]}
  \BibitemShut {NoStop}%
\bibitem [{\citenamefont {Fewster}\ and\ \citenamefont
  {Roman}(2005)}]{Fewster:2005gp}%
  \BibitemOpen
  \bibfield  {author} {\bibinfo {author} {\bibfnamefont {C.~J.}\ \bibnamefont
  {Fewster}}\ and\ \bibinfo {author} {\bibfnamefont {T.~A.}\ \bibnamefont
  {Roman}},\ }\href {\doibase 10.1103/PhysRevD.72.044023} {\bibfield  {journal}
  {\bibinfo  {journal} {Phys. Rev.}\ }\textbf {\bibinfo {volume} {D72}},\
  \bibinfo {pages} {044023} (\bibinfo {year} {2005})},\ \Eprint
  {http://arxiv.org/abs/gr-qc/0507013} {arXiv:gr-qc/0507013 [gr-qc]}
  \BibitemShut {NoStop}%
\bibitem [{\citenamefont {Penrose}(1965)}]{Penrose:1964wq}%
  \BibitemOpen
  \bibfield  {author} {\bibinfo {author} {\bibfnamefont {R.}~\bibnamefont
  {Penrose}},\ }\href {\doibase 10.1103/PhysRevLett.14.57} {\bibfield
  {journal} {\bibinfo  {journal} {Phys. Rev. Lett.}\ }\textbf {\bibinfo
  {volume} {14}},\ \bibinfo {pages} {57} (\bibinfo {year} {1965})}\BibitemShut
  {NoStop}%
\bibitem [{\citenamefont {Hawking}(1966)}]{Hawking:1966sx}%
  \BibitemOpen
  \bibfield  {author} {\bibinfo {author} {\bibfnamefont {S.}~\bibnamefont
  {Hawking}},\ }\href {\doibase 10.1098/rspa.1966.0221} {\bibfield  {journal}
  {\bibinfo  {journal} {Proc. Roy. Soc. Lond.}\ }\textbf {\bibinfo {volume}
  {A294}},\ \bibinfo {pages} {511} (\bibinfo {year} {1966})}\BibitemShut
  {NoStop}%
\bibitem [{\citenamefont {Brown}\ \emph {et~al.}(2018)\citenamefont {Brown},
  \citenamefont {Fewster},\ and\ \citenamefont {Kontou}}]{Brown:2018hym}%
  \BibitemOpen
  \bibfield  {author} {\bibinfo {author} {\bibfnamefont {P.~J.}\ \bibnamefont
  {Brown}}, \bibinfo {author} {\bibfnamefont {C.~J.}\ \bibnamefont {Fewster}},
  \ and\ \bibinfo {author} {\bibfnamefont {E.-A.}\ \bibnamefont {Kontou}},\
  }\href {\doibase 10.1007/s10714-018-2446-5} {\bibfield  {journal} {\bibinfo
  {journal} {General Relativity and Gravitation}\ }\textbf {\bibinfo {volume}
  {50}},\ \bibinfo {pages} {121} (\bibinfo {year} {2018})},\ \Eprint
  {http://arxiv.org/abs/1803.11094} {arXiv:1803.11094 [gr-qc]} \BibitemShut
  {NoStop}%
\bibitem [{\citenamefont {Pirani}(2009)}]{Pirani:2009}%
  \BibitemOpen
  \bibfield  {author} {\bibinfo {author} {\bibfnamefont {F.~A.~E.}\
  \bibnamefont {Pirani}},\ }\href {\doibase 10.1007/s10714-009-0787-9}
  {\bibfield  {journal} {\bibinfo  {journal} {Gen. Relativity Gravitation}\
  }\textbf {\bibinfo {volume} {41}},\ \bibinfo {pages} {1215} (\bibinfo {year}
  {2009})},\ \bibinfo {note} {republication of Acta Physica Polonica
  \textbf{15}, 389-–405 (1956).}\BibitemShut {Stop}%
\bibitem [{\citenamefont {Tipler}(1978)}]{Tipler:1978zz}%
  \BibitemOpen
  \bibfield  {author} {\bibinfo {author} {\bibfnamefont {F.~J.}\ \bibnamefont
  {Tipler}},\ }\href {\doibase 10.1103/PhysRevD.17.2521} {\bibfield  {journal}
  {\bibinfo  {journal} {Phys. Rev.}\ }\textbf {\bibinfo {volume} {D17}},\
  \bibinfo {pages} {2521} (\bibinfo {year} {1978})}\BibitemShut {NoStop}%
\bibitem [{\citenamefont {Chicone}\ and\ \citenamefont
  {Ehrlich}(1980)}]{chicone1980line}%
  \BibitemOpen
  \bibfield  {author} {\bibinfo {author} {\bibfnamefont {C.}~\bibnamefont
  {Chicone}}\ and\ \bibinfo {author} {\bibfnamefont {P.}~\bibnamefont
  {Ehrlich}},\ }\href@noop {} {\bibfield  {journal} {\bibinfo  {journal}
  {Manuscripta Mathematica}\ }\textbf {\bibinfo {volume} {31}},\ \bibinfo
  {pages} {297} (\bibinfo {year} {1980})}\BibitemShut {NoStop}%
\bibitem [{\citenamefont {Galloway}(1981)}]{Galloway:1981}%
  \BibitemOpen
  \bibfield  {author} {\bibinfo {author} {\bibfnamefont {G.~J.}\ \bibnamefont
  {Galloway}},\ }\href {\doibase 10.1007/BF01168457} {\bibfield  {journal}
  {\bibinfo  {journal} {Manuscripta Math.}\ }\textbf {\bibinfo {volume} {35}},\
  \bibinfo {pages} {209} (\bibinfo {year} {1981})}\BibitemShut {NoStop}%
\bibitem [{\citenamefont {Borde}(1987)}]{Borde:1987qr}%
  \BibitemOpen
  \bibfield  {author} {\bibinfo {author} {\bibfnamefont {A.}~\bibnamefont
  {Borde}},\ }\href {\doibase 10.1088/0264-9381/4/2/015} {\bibfield  {journal}
  {\bibinfo  {journal} {Class. Quant. Grav.}\ }\textbf {\bibinfo {volume}
  {4}},\ \bibinfo {pages} {343} (\bibinfo {year} {1987})}\BibitemShut {NoStop}%
\bibitem [{\citenamefont {Roman}(1988)}]{Roman:1988vv}%
  \BibitemOpen
  \bibfield  {author} {\bibinfo {author} {\bibfnamefont {T.~A.}\ \bibnamefont
  {Roman}},\ }\href {\doibase 10.1103/PhysRevD.37.546} {\bibfield  {journal}
  {\bibinfo  {journal} {Phys. Rev.}\ }\textbf {\bibinfo {volume} {D37}},\
  \bibinfo {pages} {546} (\bibinfo {year} {1988})}\BibitemShut {NoStop}%
\bibitem [{\citenamefont {Wald}\ and\ \citenamefont
  {Yurtsever}(1991)}]{Wald:1991xn}%
  \BibitemOpen
  \bibfield  {author} {\bibinfo {author} {\bibfnamefont {R.~M.}\ \bibnamefont
  {Wald}}\ and\ \bibinfo {author} {\bibfnamefont {U.}~\bibnamefont
  {Yurtsever}},\ }\href {\doibase 10.1103/PhysRevD.44.403} {\bibfield
  {journal} {\bibinfo  {journal} {Phys. Rev.}\ }\textbf {\bibinfo {volume}
  {D44}},\ \bibinfo {pages} {403} (\bibinfo {year} {1991})}\BibitemShut
  {NoStop}%
\bibitem [{\citenamefont {Borde}(1994)}]{Borde:1994ai}%
  \BibitemOpen
  \bibfield  {author} {\bibinfo {author} {\bibfnamefont {A.}~\bibnamefont
  {Borde}},\ }\href {\doibase 10.1103/PhysRevD.50.3692} {\bibfield  {journal}
  {\bibinfo  {journal} {Phys. Rev.}\ }\textbf {\bibinfo {volume} {D50}},\
  \bibinfo {pages} {3692} (\bibinfo {year} {1994})},\ \Eprint
  {http://arxiv.org/abs/gr-qc/9403049} {arXiv:gr-qc/9403049 [gr-qc]}
  \BibitemShut {NoStop}%
\bibitem [{\citenamefont {Fewster}\ and\ \citenamefont
  {Galloway}(2011)}]{Fewster:2010gm}%
  \BibitemOpen
  \bibfield  {author} {\bibinfo {author} {\bibfnamefont {C.~J.}\ \bibnamefont
  {Fewster}}\ and\ \bibinfo {author} {\bibfnamefont {G.~J.}\ \bibnamefont
  {Galloway}},\ }\href {\doibase 10.1088/0264-9381/28/12/125009} {\bibfield
  {journal} {\bibinfo  {journal} {Class. Quant. Grav.}\ }\textbf {\bibinfo
  {volume} {28}},\ \bibinfo {pages} {125009} (\bibinfo {year} {2011})},\
  \Eprint {http://arxiv.org/abs/1012.6038} {arXiv:1012.6038 [gr-qc]}
  \BibitemShut {NoStop}%
\bibitem [{\citenamefont {Lesourd}(2018)}]{Lesourd2018}%
  \BibitemOpen
  \bibfield  {author} {\bibinfo {author} {\bibfnamefont {M.}~\bibnamefont
  {Lesourd}},\ }\href {\doibase 10.1007/s10714-018-2377-1} {\bibfield
  {journal} {\bibinfo  {journal} {General Relativity and Gravitation}\ }\textbf
  {\bibinfo {volume} {50}},\ \bibinfo {pages} {61} (\bibinfo {year}
  {2018})}\BibitemShut {NoStop}%
\bibitem [{\citenamefont {Fewster}\ and\ \citenamefont
  {Osterbrink}(2008)}]{Fewster:2007ec}%
  \BibitemOpen
  \bibfield  {author} {\bibinfo {author} {\bibfnamefont {C.~J.}\ \bibnamefont
  {Fewster}}\ and\ \bibinfo {author} {\bibfnamefont {L.~W.}\ \bibnamefont
  {Osterbrink}},\ }\href {\doibase 10.1088/1751-8113/41/2/025402} {\bibfield
  {journal} {\bibinfo  {journal} {J. Phys.}\ }\textbf {\bibinfo {volume}
  {A41}},\ \bibinfo {pages} {025402} (\bibinfo {year} {2008})},\ \Eprint
  {http://arxiv.org/abs/0708.2450} {arXiv:0708.2450 [gr-qc]} \BibitemShut
  {NoStop}%
\bibitem [{\citenamefont {Ford}\ and\ \citenamefont
  {Roman}(1995)}]{Ford:1994bj}%
  \BibitemOpen
  \bibfield  {author} {\bibinfo {author} {\bibfnamefont {L.~H.}\ \bibnamefont
  {Ford}}\ and\ \bibinfo {author} {\bibfnamefont {T.~A.}\ \bibnamefont
  {Roman}},\ }\href {\doibase 10.1103/PhysRevD.51.4277} {\bibfield  {journal}
  {\bibinfo  {journal} {Phys. Rev.}\ }\textbf {\bibinfo {volume} {D51}},\
  \bibinfo {pages} {4277} (\bibinfo {year} {1995})},\ \Eprint
  {http://arxiv.org/abs/gr-qc/9410043} {arXiv:gr-qc/9410043 [gr-qc]}
  \BibitemShut {NoStop}%
\bibitem [{\citenamefont {Pfenning}\ and\ \citenamefont
  {Ford}(1998)}]{Pfenning:1997rg}%
  \BibitemOpen
  \bibfield  {author} {\bibinfo {author} {\bibfnamefont {M.~J.}\ \bibnamefont
  {Pfenning}}\ and\ \bibinfo {author} {\bibfnamefont {L.~H.}\ \bibnamefont
  {Ford}},\ }\href {\doibase 10.1103/PhysRevD.57.3489} {\bibfield  {journal}
  {\bibinfo  {journal} {Phys. Rev.}\ }\textbf {\bibinfo {volume} {D57}},\
  \bibinfo {pages} {3489} (\bibinfo {year} {1998})},\ \Eprint
  {http://arxiv.org/abs/gr-qc/9710055} {arXiv:gr-qc/9710055 [gr-qc]}
  \BibitemShut {NoStop}%
\bibitem [{\citenamefont {Fewster}(2000)}]{Fewster:1999gj}%
  \BibitemOpen
  \bibfield  {author} {\bibinfo {author} {\bibfnamefont {C.~J.}\ \bibnamefont
  {Fewster}},\ }\href {\doibase 10.1088/0264-9381/17/9/302} {\bibfield
  {journal} {\bibinfo  {journal} {Class. Quant. Grav.}\ }\textbf {\bibinfo
  {volume} {17}},\ \bibinfo {pages} {1897} (\bibinfo {year} {2000})},\ \Eprint
  {http://arxiv.org/abs/gr-qc/9910060} {arXiv:gr-qc/9910060 [gr-qc]}
  \BibitemShut {NoStop}%
\bibitem [{\citenamefont {Flanagan}(1997)}]{Flanagan:1997gn}%
  \BibitemOpen
  \bibfield  {author} {\bibinfo {author} {\bibfnamefont {E.~E.}\ \bibnamefont
  {Flanagan}},\ }\href {\doibase 10.1103/PhysRevD.56.4922} {\bibfield
  {journal} {\bibinfo  {journal} {Phys. Rev.}\ }\textbf {\bibinfo {volume}
  {D56}},\ \bibinfo {pages} {4922} (\bibinfo {year} {1997})},\ \Eprint
  {http://arxiv.org/abs/gr-qc/9706006} {arXiv:gr-qc/9706006 [gr-qc]}
  \BibitemShut {NoStop}%
\bibitem [{\citenamefont {Flanagan}(2002)}]{Flanagan:2002bd}%
  \BibitemOpen
  \bibfield  {author} {\bibinfo {author} {\bibfnamefont {E.~E.}\ \bibnamefont
  {Flanagan}},\ }\href {\doibase 10.1103/PhysRevD.66.104007} {\bibfield
  {journal} {\bibinfo  {journal} {Phys. Rev.}\ }\textbf {\bibinfo {volume}
  {D66}},\ \bibinfo {pages} {104007} (\bibinfo {year} {2002})},\ \Eprint
  {http://arxiv.org/abs/gr-qc/0208066} {arXiv:gr-qc/0208066 [gr-qc]}
  \BibitemShut {NoStop}%
\bibitem [{\citenamefont {Fewster}\ and\ \citenamefont
  {Smith}(2008)}]{Fewster:2007rh}%
  \BibitemOpen
  \bibfield  {author} {\bibinfo {author} {\bibfnamefont {C.~J.}\ \bibnamefont
  {Fewster}}\ and\ \bibinfo {author} {\bibfnamefont {C.~J.}\ \bibnamefont
  {Smith}},\ }\href {\doibase 10.1007/s00023-008-0361-0} {\bibfield  {journal}
  {\bibinfo  {journal} {Annales Henri Poincare}\ }\textbf {\bibinfo {volume}
  {9}},\ \bibinfo {pages} {425} (\bibinfo {year} {2008})},\ \Eprint
  {http://arxiv.org/abs/gr-qc/0702056} {arXiv:gr-qc/0702056 [GR-QC]}
  \BibitemShut {NoStop}%
\bibitem [{\citenamefont {Kontou}\ and\ \citenamefont
  {Olum}(2015)}]{Kontou:2014tha}%
  \BibitemOpen
  \bibfield  {author} {\bibinfo {author} {\bibfnamefont {E.-A.}\ \bibnamefont
  {Kontou}}\ and\ \bibinfo {author} {\bibfnamefont {K.~D.}\ \bibnamefont
  {Olum}},\ }\href {\doibase 10.1103/PhysRevD.91.104005} {\bibfield  {journal}
  {\bibinfo  {journal} {Phys. Rev.}\ }\textbf {\bibinfo {volume} {D91}},\
  \bibinfo {pages} {104005} (\bibinfo {year} {2015})},\ \Eprint
  {http://arxiv.org/abs/1410.0665} {arXiv:1410.0665 [gr-qc]} \BibitemShut
  {NoStop}%
\bibitem [{\citenamefont {Fewster}(2007)}]{Fewster:2006iy}%
  \BibitemOpen
  \bibfield  {author} {\bibinfo {author} {\bibfnamefont {C.~J.}\ \bibnamefont
  {Fewster}},\ }\href {\doibase 10.1007/s10714-007-0494-3} {\bibfield
  {journal} {\bibinfo  {journal} {Gen. Rel. Grav.}\ }\textbf {\bibinfo {volume}
  {39}},\ \bibinfo {pages} {1855} (\bibinfo {year} {2007})},\ \Eprint
  {http://arxiv.org/abs/math-ph/0611058} {arXiv:math-ph/0611058 [math-ph]}
  \BibitemShut {NoStop}%
\bibitem [{\citenamefont {Misner}\ \emph {et~al.}(1973)\citenamefont {Misner},
  \citenamefont {Thorne},\ and\ \citenamefont {Wheeler}}]{MTW}%
  \BibitemOpen
  \bibfield  {author} {\bibinfo {author} {\bibfnamefont {C.~W.}\ \bibnamefont
  {Misner}}, \bibinfo {author} {\bibfnamefont {K.}~\bibnamefont {Thorne}}, \
  and\ \bibinfo {author} {\bibfnamefont {J.}~\bibnamefont {Wheeler}},\
  }\href@noop {} {\emph {\bibinfo {title} {Gravitation}}}\ (\bibinfo
  {publisher} {W. H. Freeman},\ \bibinfo {address} {San Francisco},\ \bibinfo
  {year} {1973})\BibitemShut {NoStop}%
\bibitem [{\citenamefont {Khavkine}\ and\ \citenamefont
  {Moretti}(2015)}]{KhavkineMoretti-aqft}%
  \BibitemOpen
  \bibfield  {author} {\bibinfo {author} {\bibfnamefont {I.}~\bibnamefont
  {Khavkine}}\ and\ \bibinfo {author} {\bibfnamefont {V.}~\bibnamefont
  {Moretti}},\ }in\ \href@noop {} {\emph {\bibinfo {booktitle} {Advances in
  algebraic quantum field theory}}},\ \bibinfo {series and number} {Math. Phys.
  Stud.}\ (\bibinfo  {publisher} {Springer, Cham},\ \bibinfo {year} {2015})\
  pp.\ \bibinfo {pages} {191--251}\BibitemShut {NoStop}%
\bibitem [{\citenamefont {{Fewster}}(2018)}]{Fewster_artofstate:2018}%
  \BibitemOpen
  \bibfield  {author} {\bibinfo {author} {\bibfnamefont {C.~J.}\ \bibnamefont
  {{Fewster}}},\ }\href {\doibase 10.1142/S0218271818430071} {\bibfield
  {journal} {\bibinfo  {journal} {International Journal of Modern Physics D}\
  }\textbf {\bibinfo {volume} {27}},\ \bibinfo {eid} {1843007-251} (\bibinfo
  {year} {2018})},\ \Eprint {http://arxiv.org/abs/1803.06836} {arXiv:1803.06836
  [gr-qc]} \BibitemShut {NoStop}%
\bibitem [{\citenamefont {Kay}\ and\ \citenamefont
  {Wald}(1991)}]{KayWald:1991}%
  \BibitemOpen
  \bibfield  {author} {\bibinfo {author} {\bibfnamefont {B.~S.}\ \bibnamefont
  {Kay}}\ and\ \bibinfo {author} {\bibfnamefont {R.~M.}\ \bibnamefont {Wald}},\
  }\href {\doibase 10.1016/0370-1573(91)90015-E} {\bibfield  {journal}
  {\bibinfo  {journal} {Phys. Rep.}\ }\textbf {\bibinfo {volume} {207}},\
  \bibinfo {pages} {49} (\bibinfo {year} {1991})}\BibitemShut {NoStop}%
\bibitem [{\citenamefont {Radzikowski}(1996)}]{Radzikowski:1996}%
  \BibitemOpen
  \bibfield  {author} {\bibinfo {author} {\bibfnamefont {M.~J.}\ \bibnamefont
  {Radzikowski}},\ }\href {http://projecteuclid.org/euclid.cmp/1104287114}
  {\bibfield  {journal} {\bibinfo  {journal} {Comm. Math. Phys.}\ }\textbf
  {\bibinfo {volume} {179}},\ \bibinfo {pages} {529} (\bibinfo {year}
  {1996})}\BibitemShut {NoStop}%
\bibitem [{\citenamefont {Duistermaat}(2011)}]{Duistermaat_FIO}%
  \BibitemOpen
  \bibfield  {author} {\bibinfo {author} {\bibfnamefont {J.~J.}\ \bibnamefont
  {Duistermaat}},\ }\href {\doibase 10.1007/978-0-8176-8108-1} {\emph {\bibinfo
  {title} {Fourier integral operators}}},\ Modern Birkh\"auser Classics\
  (\bibinfo  {publisher} {Birkh\"auser/Springer, New York},\ \bibinfo {year}
  {2011})\BibitemShut {NoStop}%
\bibitem [{\citenamefont {Brouder}\ \emph {et~al.}(2014)\citenamefont
  {Brouder}, \citenamefont {Dang},\ and\ \citenamefont
  {H\'elein}}]{Brouder_etal:2014}%
  \BibitemOpen
  \bibfield  {author} {\bibinfo {author} {\bibfnamefont {C.}~\bibnamefont
  {Brouder}}, \bibinfo {author} {\bibfnamefont {N.~V.}\ \bibnamefont {Dang}}, \
  and\ \bibinfo {author} {\bibfnamefont {F.}~\bibnamefont {H\'elein}},\ }\href
  {\doibase 10.1088/1751-8113/47/44/443001} {\bibfield  {journal} {\bibinfo
  {journal} {J. Phys. A}\ }\textbf {\bibinfo {volume} {47}},\ \bibinfo {pages}
  {443001, 30} (\bibinfo {year} {2014})}\BibitemShut {NoStop}%
\bibitem [{\citenamefont {Hollands}\ and\ \citenamefont
  {Wald}(2001)}]{Hollands:2001nf}%
  \BibitemOpen
  \bibfield  {author} {\bibinfo {author} {\bibfnamefont {S.}~\bibnamefont
  {Hollands}}\ and\ \bibinfo {author} {\bibfnamefont {R.~M.}\ \bibnamefont
  {Wald}},\ }\href {\doibase 10.1007/s002200100540} {\bibfield  {journal}
  {\bibinfo  {journal} {Commun. Math. Phys.}\ }\textbf {\bibinfo {volume}
  {223}},\ \bibinfo {pages} {289} (\bibinfo {year} {2001})},\ \Eprint
  {http://arxiv.org/abs/gr-qc/0103074} {arXiv:gr-qc/0103074 [gr-qc]}
  \BibitemShut {NoStop}%
\bibitem [{\citenamefont {Hollands}\ and\ \citenamefont
  {Wald}(2005)}]{Hollands:2004yh}%
  \BibitemOpen
  \bibfield  {author} {\bibinfo {author} {\bibfnamefont {S.}~\bibnamefont
  {Hollands}}\ and\ \bibinfo {author} {\bibfnamefont {R.~M.}\ \bibnamefont
  {Wald}},\ }\href {\doibase 10.1142/S0129055X05002340} {\bibfield  {journal}
  {\bibinfo  {journal} {Rev. Math. Phys.}\ }\textbf {\bibinfo {volume} {17}},\
  \bibinfo {pages} {227} (\bibinfo {year} {2005})},\ \Eprint
  {http://arxiv.org/abs/gr-qc/0404074} {arXiv:gr-qc/0404074 [gr-qc]}
  \BibitemShut {NoStop}%
\bibitem [{\citenamefont {D\"utsch}\ and\ \citenamefont
  {Fredenhagen}(2004)}]{DutschFredenhagen:2004}%
  \BibitemOpen
  \bibfield  {author} {\bibinfo {author} {\bibfnamefont {M.}~\bibnamefont
  {D\"utsch}}\ and\ \bibinfo {author} {\bibfnamefont {K.}~\bibnamefont
  {Fredenhagen}},\ }\href {\doibase 10.1142/S0129055X04002266} {\bibfield
  {journal} {\bibinfo  {journal} {Rev. Math. Phys.}\ }\textbf {\bibinfo
  {volume} {16}},\ \bibinfo {pages} {1291} (\bibinfo {year}
  {2004})}\BibitemShut {NoStop}%
\bibitem [{\citenamefont {Moretti}(2003)}]{Moretti:2001qh}%
  \BibitemOpen
  \bibfield  {author} {\bibinfo {author} {\bibfnamefont {V.}~\bibnamefont
  {Moretti}},\ }\href {\doibase 10.1007/s00220-002-0702-7} {\bibfield
  {journal} {\bibinfo  {journal} {Commun. Math. Phys.}\ }\textbf {\bibinfo
  {volume} {232}},\ \bibinfo {pages} {189} (\bibinfo {year} {2003})},\ \Eprint
  {http://arxiv.org/abs/gr-qc/0109048} {arXiv:gr-qc/0109048 [gr-qc]}
  \BibitemShut {NoStop}%
\bibitem [{\citenamefont {Wald}(1995)}]{Wald:1995yp}%
  \BibitemOpen
  \bibfield  {author} {\bibinfo {author} {\bibfnamefont {R.~M.}\ \bibnamefont
  {Wald}},\ }\href@noop {} {\emph {\bibinfo {title} {{Quantum Field Theory in
  Curved Space-Time and Black Hole Thermodynamics}}}},\ Chicago Lectures in
  Physics\ (\bibinfo  {publisher} {University of Chicago Press},\ \bibinfo
  {address} {Chicago, IL},\ \bibinfo {year} {1995})\BibitemShut {NoStop}%
\bibitem [{\citenamefont {Glaeser}(1963)}]{Glaeser:1963}%
  \BibitemOpen
  \bibfield  {author} {\bibinfo {author} {\bibfnamefont {G.}~\bibnamefont
  {Glaeser}},\ }\href {http://www.numdam.org/item?id=AIF_1963__13_2_203_0}
  {\bibfield  {journal} {\bibinfo  {journal} {Ann. Inst. Fourier (Grenoble)}\
  }\textbf {\bibinfo {volume} {13}},\ \bibinfo {pages} {203} (\bibinfo {year}
  {1963})}\BibitemShut {NoStop}%
\bibitem [{\citenamefont {Bony}\ \emph {et~al.}(2010)\citenamefont {Bony},
  \citenamefont {Colombini},\ and\ \citenamefont
  {Pernazza}}]{BonyColombiniPernazza:2010}%
  \BibitemOpen
  \bibfield  {author} {\bibinfo {author} {\bibfnamefont {J.-M.}\ \bibnamefont
  {Bony}}, \bibinfo {author} {\bibfnamefont {F.}~\bibnamefont {Colombini}}, \
  and\ \bibinfo {author} {\bibfnamefont {L.}~\bibnamefont {Pernazza}},\
  }\href@noop {} {\bibfield  {journal} {\bibinfo  {journal} {Ann. Sc. Norm.
  Super. Pisa Cl. Sci. (5)}\ }\textbf {\bibinfo {volume} {9}},\ \bibinfo
  {pages} {635} (\bibinfo {year} {2010})}\BibitemShut {NoStop}%
\bibitem [{\citenamefont {Ford}(1991)}]{Ford:1990id}%
  \BibitemOpen
  \bibfield  {author} {\bibinfo {author} {\bibfnamefont {L.~H.}\ \bibnamefont
  {Ford}},\ }\href {\doibase 10.1103/PhysRevD.43.3972} {\bibfield  {journal}
  {\bibinfo  {journal} {Phys. Rev.}\ }\textbf {\bibinfo {volume} {D43}},\
  \bibinfo {pages} {3972} (\bibinfo {year} {1991})}\BibitemShut {NoStop}%
\bibitem [{\citenamefont {Ford}\ and\ \citenamefont
  {Roman}(1997)}]{Ford:1996er}%
  \BibitemOpen
  \bibfield  {author} {\bibinfo {author} {\bibfnamefont {L.~H.}\ \bibnamefont
  {Ford}}\ and\ \bibinfo {author} {\bibfnamefont {T.~A.}\ \bibnamefont
  {Roman}},\ }\href {\doibase 10.1103/PhysRevD.55.2082} {\bibfield  {journal}
  {\bibinfo  {journal} {Phys. Rev.}\ }\textbf {\bibinfo {volume} {D55}},\
  \bibinfo {pages} {2082} (\bibinfo {year} {1997})},\ \Eprint
  {http://arxiv.org/abs/gr-qc/9607003} {arXiv:gr-qc/9607003 [gr-qc]}
  \BibitemShut {NoStop}%
\bibitem [{\citenamefont {Fewster}\ and\ \citenamefont
  {Eveson}(1998)}]{Fewster:1998pu}%
  \BibitemOpen
  \bibfield  {author} {\bibinfo {author} {\bibfnamefont {C.~J.}\ \bibnamefont
  {Fewster}}\ and\ \bibinfo {author} {\bibfnamefont {S.~P.}\ \bibnamefont
  {Eveson}},\ }\href {\doibase 10.1103/PhysRevD.58.084010} {\bibfield
  {journal} {\bibinfo  {journal} {Phys. Rev.}\ }\textbf {\bibinfo {volume}
  {D58}},\ \bibinfo {pages} {084010} (\bibinfo {year} {1998})},\ \Eprint
  {http://arxiv.org/abs/gr-qc/9805024} {arXiv:gr-qc/9805024 [gr-qc]}
  \BibitemShut {NoStop}%
\end{thebibliography}%

\end{document}